\documentclass[]{elsarticle}
\usepackage[utf8]{inputenc}

\usepackage{amsmath}
\usepackage{amssymb}
\usepackage{bbm}                            
\usepackage{csquotes}                       
\usepackage{color}                          
\usepackage{overpic}			    
\usepackage{algorithm}
\usepackage{algpseudocode}                  

\newtheorem{thm}{Theorem}
\newtheorem{lem}[thm]{Lemma}
\newtheorem{prop}[thm]{Proposition}
\newdefinition{defi}[thm]{Definition}
\newdefinition{rmk}[thm]{Remark}
\newdefinition{algo}[thm]{Algorithm}
\newproof{proof}{Proof}

\newcommand{\arc}{\ensuremath{\gamma}}
\newcommand{\arceins}{\ensuremath{\gamma_1}}
\newcommand{\arczwei}{\ensuremath{\gamma_2}}
\newcommand{\arcdrei}{\ensuremath{\gamma_3}}

\newcommand{\RR}{\mathbbm{R}}               
\newcommand{\ZZ}{\mathbbm{Z}}               
\newcommand{\NN}{\mathbbm{N}}               

\newcommand{\capl}{\cap_-}
\newcommand{\capr}{\cap_+}

\newcommand{\restr}[1]{\ensuremath{\big|_{#1}}}

\newcommand{\reverse}[1]{
\hbox{%
    \vbox{%
      \hrule height 0.5pt 
      \kern0.5ex
      \hbox{\ensuremath{#1}}
    }%
  }%
}

\newcommand{\normal}[2]{\ensuremath{#1^{\perp}(#2)}}

\newcommand{\sprod}[2]{\ensuremath{\langle \, #1 \,,\, #2 \, \rangle}}

\newcommand{\fdg}{\,:\enspace}

\journal{Computer Aided Geometric Design}

\begin{document}

\begin{frontmatter}

\title{Numerically robust computation of circular visibility}
 \author[label1]{Stephan Brummer\corref{cor1}}
 \ead{stephan.brummer@uni-passau.de}
 \author[label1]{Georg Maier}
 \author[label1,label2]{Tomas Sauer}
 \cortext[cor1]{Corresponding author}
 \address[label1]{FORWISS, Universität Passau, Innstr.~43, 94032 Passau, Germany}
 \address[label2]{Lehrstuhl für Mathematik mit Schwerpunkt Digitale Signalverarbeitung,
Universität Passau, Innstr.~43, 94032 Passau, Germany}

\begin{abstract}
We address the question of whether a point inside a domain bounded by a
simple closed arc spline is circularly visible from a specified
arc from the boundary. We provide a simple and numerically stable
linear time algorithm that solves this problem.
In particular, we present an easy-to-check criterion that implies
that a point is not visible from a specified boundary arc.
\end{abstract}


\end{frontmatter}

\section{Introduction}
A point in the plane is called \emph{circularly visible} from another point inside a planar domain
if the two points can be connected by a circular arc that lies inside this domain.
Algorithms that compute the set of all circularly visible points inside a polygon
from a point or edge are well studied, cf.~\cite{agarwal1993, chou1992, chou1995}.
In~\cite{agarwal1993}, circular visibility from a point inside
a simple polygon is treated. The authors present an $O(n \log n)$ algorithm,
with $n$ the number of vertices of the polygon, that
computes the set of all circularly visible points.
In~\cite{chou1995} an algorithm to compute the circular visibility set
of a point inside a simple polygon is presented, which is based on the so called
CVD (circular visibility diagram), a partition of the plane where every
point represents the center of an arc. This leads to an algorithm
with linear runtime with respect to the number of vertices. A discussion of
numerical stability is not existing in both cases,
but numerical problems can be assumed if relevant circular arcs are almost straight.
The computation of the CVD is further used to compute the circular
visibility set from an edge of a simple polygon in~\cite{chou1992}.
The runtime of the presented algorithm is $O(kn)$ where $n$ is the number of
vertices and $k$ is the number of CVDs computed which equals $n$ in worst case.
The problem we tackle here differs in two ways:
we want to consider domains bounded by an \emph{arc spline}, a curve
that consists of circular arcs and line segments, and we
only want to know if a point is visible from a specified arc on
the boundary.
The treatment of regions bounded by arc splines has not been considered yet in literature.
We present a simple and numerically stable algorithm that decides,
in linear time with respect to the number of arc segments,
if a point is circularly visible from a boundary arc.
For this purpose, we supply an easy-to-check criterion that directly implies
that a point is not circularly visible from an arc.
Although we only consider circular visibility of a point,
we compute in some sense extremal arcs having a so-called alternating sequence.
This enables that this approach can easily be extended to compute boundary arcs
of the circular visibility set. This is a nice property as you are usually
most interested in this boundary region.\\
We use this algorithm to improve the numerical stability of the
SMAP (smooth minimum arc path) approach which computes an approximating smooth arc spline
with the minimal number of segments within a specified
maximal tolerance, cf.~\cite{Maier2014} or \cite{schindler2013} for an application
in vehicle self-localization. The basic task in the SMAP algorithm
is closely related to the computation of the circular visibility set from a starting arc.
It is known that the boundary of the circular visibility set consists
of \enquote{boundary arcs} having three points in common with the boundary
of the domain. Due to even very small numerical inaccuracies, however,
such boundary arcs can be missed.
With the approach presented in this paper, we can determine
if a point is visible and so we can localize the position of boundary arcs.\\
This paper is organized as follows: In Section \ref{section_notation}, we introduce
basic notations and definitions. In Section \ref{section_order}, we define a key tool for later
proofs, a total order on a specified set of arcs.
In Section \ref{section_restrictions}, a sufficient condition for a point to be not circularly visible
from an arc is shown. In Section \ref{section_algorithm}, we present a linear time
algorithm to decide if a point is circularly visible from an arc.
\section{Notation and basic definitions\label{section_notation}}
We call a continuous mapping $\alpha : [0,1] \to \RR^2$ a \emph{path}
and $\alpha(0)$ its \emph{starting point} and $\alpha(1)$ its \emph{endpoint}.
A path $\alpha$ is \emph{closed} if $\alpha(0) = \alpha(1)$, it is \emph{simple}
if it is injective and \emph{simple closed} if it is closed
and $\alpha\restr{[0,1)}$ is injective.
Note that the image of a simple closed path is a \emph{Jordan curve}
which divides its complement, according to the Jordan curve theorem,
into two connected components: a bounded one which we call the
\emph{interior} of the Jordan curve and an unbounded one, its \emph{exterior}.
As usual in the literature, we will use $\alpha$ for both the mapping and the image,
usually referred to as a \emph{curve}.
In particular, this allows us to write $p\in\alpha$ instead of $p \in \alpha([0,1])$.\\
We denote $\alpha((0,1))$ by $\alpha^\circ$
and by $\reverse{\alpha}$ the \emph{reverse path} defined by $\reverse{\alpha}(t) = \alpha(1-t)$.
Let $\alpha$ be a simple path and $p \in \alpha$.
We denote by $t_\alpha(p)$ the unique parameter in $[0,1]$ with $\alpha(t_\alpha(p)) = p$.
We write $t(p)$ if the corresponding path is clear from the context.
For $p,q\in\alpha$ we write $p \prec_\alpha q$ if
$t_\alpha(p) < t_\alpha(q)$.\\
A path $\arc$ of the form
\[
\arc(t) =
c + r \cdot \begin{pmatrix} \cos(2\pi a t + t_1) \\ \sin(2\pi a t + t_1) \end{pmatrix}, \qquad c \in \RR^2, r > 0, a \in (0,1), t_1 \in [0,2\pi),
\]
is called a \emph{positively oriented arc}. We call the reverse path $\reverse{\arc}$
of a positively oriented arc a \emph{negatively oriented arc}.
The path $\ell$ defined by
$\ell(t) = (1-t) \cdot p_1 + t \cdot p_2,\, p_1,p_2 \in \RR^2, p_1 \not= p_2$,
is a \emph{line segment from $p_1$ to $p_2$} denoted by $[p_1,p_2]$.
We call a path an \emph{arc} if it is an arc of either orientation or a line segment.
The set of all arcs will be denoted by $\Gamma$.\\
As an arc $\arc$ is differentiable with respect to $t$ and its derivative $\dot{\arc}(t)$
does not vanish for any $t \in [0,1]$, we can define the
\emph{unit tangent vector} $\arc' : [0,1] \to S^1$, where $S^1$ is
the unit sphere, by $\arc'(t) := \frac{\dot{\arc}(t)}{||\dot{\arc}(t)||_2}$
and the \emph{normal of length one \enquote{to the left}}
$\normal{\arc}{t} := (-v, u)^T$ with $(u,v)^T = \arc'(t)$.\\
For $p,q,r \in \RR^2, \tau \in S^1$, we denote
by $\arc[p,r,q]$ the arc with starting point $p$, endpoint $q$
that passes through $r$ and by $\arc[\tau,p,q]$
the arc with starting point $p$, endpoint $q$
and $\tau$ as starting point tangent and by $\arc[p,q,\tau]$ the arc
with starting point $p$, endpoint $q$ and $\tau$ as endpoint tangent.
Note that $\arc[p,r,q]$ exists and is unique if $p,q,r$ are distinct,
$q \not\in [p,r]$ and $p \not\in [r,q]$.
Likewise, $\arc[\tau,p,q]$ and $\arc[p,q,\tau]$
exist and are unique if $p\not= q$ and $\tau$ and $(p-q)$
are not pointing into the same direction.\\
Let $\arc$ be a positively or negatively oriented arc or a line segment,
then we call $[\arc] := \arc(\RR)$ the
\emph{corresponding circle} or the \emph{corresponding line}, respectively.\\
Let $\arc_1, \arc_2, \ldots, \arc_n$ be arcs with $\arc_k(1) = \arc_{k+1}(0), \, k \in \{1, \ldots n-1\}$.
Then, we call the path $\arc_1 \sqcup \arc_2 \sqcup \ldots \sqcup \arc_n$
defined as the concatenation
\[
(\arc_1 \sqcup \arc_2 \sqcup \cdots \sqcup \arc_n)(t) := \arc_k(nt - k + 1), \,\, t \in \frac{1}{n}[k-1,k], \,\, k = 1,2, \ldots, n
\]
an \emph{arc spline with $n$ segments}.
We call an arc spline \emph{simple}, \emph{closed} or \emph{simple closed} if the corresponding path
is simple, closed or simple closed, respectively.
The points $\gamma_1(0), \gamma_2(0), \ldots, \gamma_n(0), \gamma_n(1)$ are called the \emph{breakpoints}
of the arc spline.\\
Let $\ell$ be a line segment and $p \in \RR^2$.
A point $p$ is \emph{strictly left of} $\ell$ if \linebreak$\sprod{\normal{\ell}{0}}{p - \ell(0)} > 0$
and it is \emph{strictly right of} $\ell$ if the inner product is negative.
We say that $p$ is \emph{strictly left of} a positively oriented arc
$\arc$ if $p$ is in the interior of $[\arc]$,
it is \emph{strictly left of} a negatively oriented arc
$\arc$ if $p$ is in the exterior of $[\arc]$.
Furthermore, $p$ is left of an arc $\gamma$ if it is either strictly left of $\gamma$ or $p \in [\gamma]$.
With $p \in \gamma^\circ$, a set $M \subset \RR^2$ is said to be \emph{locally left of $\gamma$ at $p$}
if there is an $\varepsilon > 0$ so that for every $\delta \in (0,\varepsilon)$
the set $M \cap B_p(\delta)$, with $B_p(\delta) := \{x \in \RR^2 \fdg ||x - p||_2 < \delta\}$,
is nonempty and every $q \in M \cap B_p(\delta)$ is
left of $\gamma$. We say that $M$ is \emph{locally left of $\gamma$} if
for every $p \in \gamma^\circ$ it is locally left of $\gamma$ at $p$.\\
Let $\arc$ be an arc, $\alpha$ a path and $t\in[0,1]$ with
$\alpha(t) \in \arc^\circ$.
We say \emph{$\alpha$ leaves $\arc$ in $t$ to the left}
if $\alpha(t+\varepsilon)$ is strictly left of $\arc$
for every sufficiently small $\varepsilon > 0$.
Likewise, we say \emph{$\alpha$ approaches $\arc$ in $t$ from the left}
if $\alpha(t-\varepsilon)$ is strictly left of $\arc$
for every sufficiently small $\varepsilon > 0$.
Likewise, the definitions hold for \enquote{right} instead of \enquote{left}.\\
We say that \emph{$\alpha$ cuts $\arc$ in $t$ from the left} if
there is a $t' \in [0,t]$ with $\alpha([t',t]) \in \arc$ such that
$\alpha$ approaches $\arc$ in $t'$ from the left
and it leaves $\arc$ in $t$ to the right. Likewise, we define a \emph{cut from the right}.
Note that if $\alpha$ is an arc then $\alpha$ cuts $\arc$ in $t$ from the right
if and only if $\sprod{\normal{\arc}{t_\arc(\alpha(t))}}{\alpha'(t)} > 0$.
Hence, this definition is consistent with the usual intuition of
cutting from the left or right. Furthermore, this yields that an arc $\arc_1$ cuts
another arc $\arc_2$ from the left if and only if $\arc_2$ cuts $\arc_1$ from the right.
We use
\begin{align*}
\alpha \capl \arc &:= \{t\in[0,1] \fdg \alpha \text{ cuts } \arc \text{ in } t \text{ from the left }\},\\
\alpha \capr \arc &:= \{t\in[0,1] \fdg \alpha \text{ cuts } \arc \text{ in } t \text{ from the right }\}
\end{align*}
\begin{defi}[Channel\label{defi_channel}]
Let $\sigma$ be an arc
and $\kappa = \kappa_1 \sqcup \kappa_2 \sqcup \ldots \sqcup \kappa_n$ an arc spline
with arcs $\kappa_1, \ldots, \kappa_n$ such that $\sigma \sqcup \kappa$ is simple closed.
Furthermore, we demand that the closure of the interior of $\sigma \sqcup \kappa$,
denoted by $P$, is locally left of $\sigma$.
We call $\sigma$ the \emph{starting arc} of the \emph{channel} $P$,
$\kappa$ is called the \emph{channel boundary} and $\kappa_j, j\in\{1, \ldots, n\}$
are called \emph{channel segments}.
\end{defi}
Channels are a standard tool to approximate data
within a certain maximal tolerance. It is usual to use polygonal channels
and there are efficient methods to construct them, cf.~\cite{held2008}.
Here, however, we use arc splines as channel boundary as this is necessary
in the SMAP approach, cf.~\cite{Maier2014}.
\begin{figure}[t]
\begin{minipage}[b]{0.47\textwidth}
\begin{center}
\begin{overpic}[width=3.2cm,angle=33]{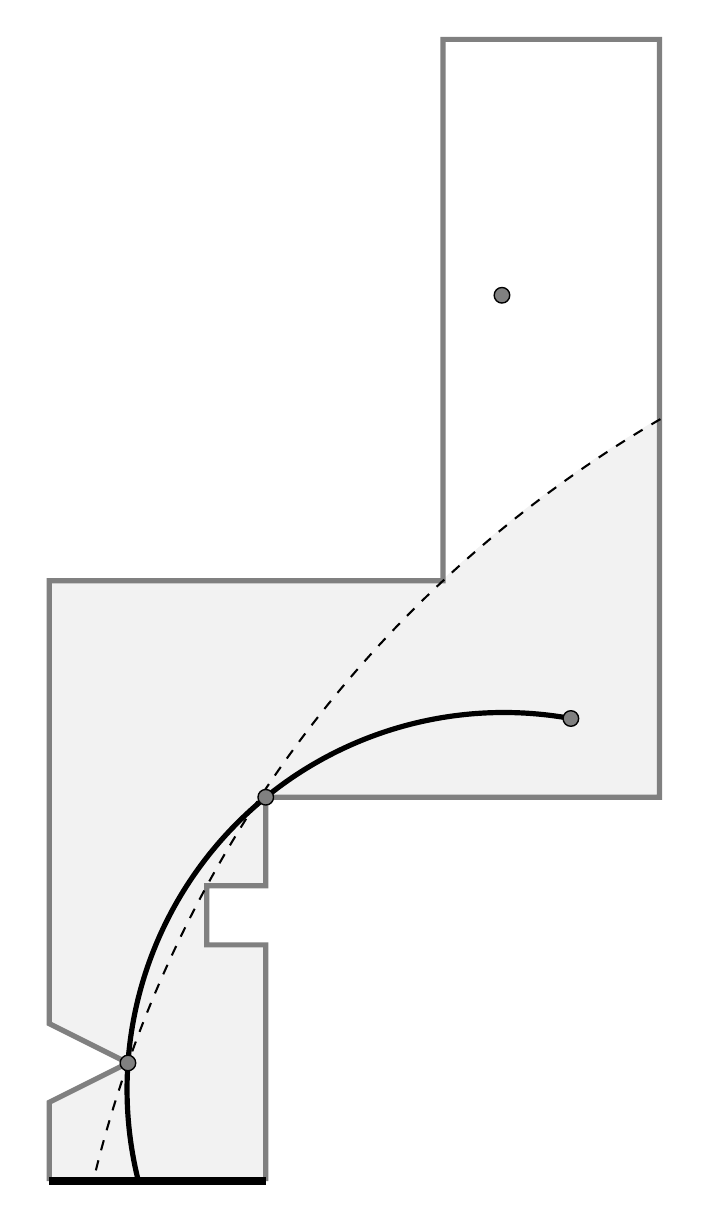}
\thicklines%
\put(33,84){$P$}
\put(55,5){$\sigma$}
\put(59,54){$p$}
\put(43,72){$q$}
\put(40,27){$\gamma$}
\end{overpic}
\caption{Channel $P$ with starting arc $\sigma$; set of circularly visible points shaded;
visible point $p$ with visibility arc $\gamma$; $q$ is not visible \label{figure_channel_visibility_set}}
\end{center}
\end{minipage}
\hfill
\begin{minipage}[b]{0.47\textwidth}
\begin{center}
\begin{overpic}[width=4cm]{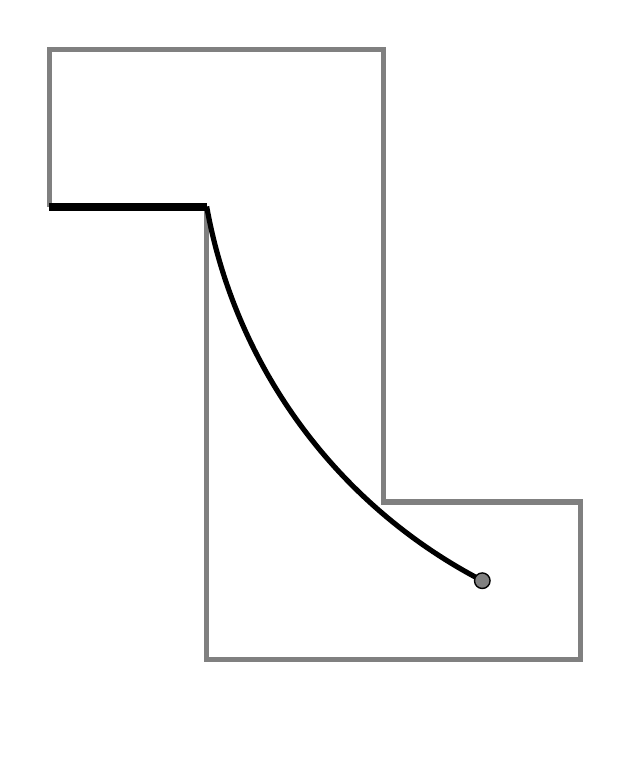}
\put(41,84){$P$}
\put(15,66){$\sigma$}
\put(64,23){$p$}
\put(36,35){$\gamma$}
\end{overpic}
\caption{A channel violating the condition of Remark~\ref{remark_convex_start_corner}
and a visibility arc $\gamma$ leaving $\sigma$ to the right\label{figure_invalid_channel}}
\end{center}
\end{minipage}
\end{figure}
\begin{defi}[Circular Visibility]
We say that a point $p \in P^\circ$ is \emph{(circularly) visible} if there
is an arc $\gamma$ with starting point on $\sigma$, endpoint $p$ and $\gamma \subset P$.
In this case, we call $\gamma$ a \emph{visibility arc}.
\end{defi}
The goal of the paper is an algorithm that either computes a visibility arc or
proves that the point is not visible.
An example showing a channel and a visibility arc is depicted in
Figure~\ref{figure_channel_visibility_set}.
\begin{rmk}\label{remark_convex_start_corner}
Throughout this paper we will assume that the channel satisfies
$\sprod{\normal{\sigma}{1}}{\kappa'(0)} > 0$ and $\sprod{\normal{\sigma}{0}}{\kappa'(1)} < 0$.
This assumption is actually not necessary for this approach but it
simplifies some of the following presentation.
One benefit is that every arc in $P$ starting at $\sigma$ has to leave
$\sigma$ to the left which reduces the complexity of some proofs.
In Figure~\ref{figure_invalid_channel} a channel
we want to omit and a visibility arc starting at $\sigma$ to the right is shown.
\end{rmk}
\section{Connecting arcs and total order\label{section_order}}
We know that a point $p$ is circularly visible from the starting arc $\sigma$
if there exists a visibility arc, that is, an arc inside the channel with
starting point on $\sigma$ and endpoint $p$.
In this chapter, we will study the set of candidates for visibility arcs,
the so-called connecting arcs.
These are arcs with starting point on $\sigma$ and endpoint $p$, but
we do not yet care if they are inside the channel.
We will define a total order on this set
of arcs which will provide a helpful tool to find a visibility arc
and to prove a criterion that classifies a point as not visible.\\
In this chapter we assume that $\sigma$ is an arc and $p \in \RR^2 \setminus \sigma$.
\begin{defi}[Connecting Arc]\label{defi_connecting_arc}
Let $\sigma$ be an arc and $p \in \RR^2 \setminus \sigma$.
If $p\not\in[\sigma]$ then let
\begin{align*}
\Gamma^*(\sigma, p) := \{ \arc \in \Gamma \fdg \arc(0) &\in \sigma^\circ, \arc \cap \sigma = \{\arc(0)\}, \arc(1) = p\\
& \text{ and } \arc \text{ leaves } \sigma \text{ in } 0 \text{ to the left } \}
\end{align*}
and
$\Gamma(\sigma,p) := \overline{\Gamma^*(\sigma,p)}$.
Otherwise, let
\begin{align*}
\Gamma(\sigma,p) :=& \{ \arc[\tau, \sigma(t), p] \fdg t\in[0,1], \tau \in S^1, \sprod{\normal{\sigma}{t}}{\tau} > 0\}\\
&\cup\, \{\;\arc[-\sigma'(0),\sigma(0),p],\; \arc[\sigma'(1),\sigma(1),p]\;\}
\end{align*}
We call an arc $\arc \in \Gamma(\sigma, p)$
a \emph{connecting arc (from $\sigma$ to $p$)}.\\
Let $q,r\in\RR^2 \setminus \sigma$ and $q,r$ and $p$ pairwise distinct.
If there exists a unique $\arc \in \Gamma(\sigma,p)$ with $q,r \in \gamma$,
$q \prec_\arc r$, then we denote $\arc$ by $\arc_\sigma[q,r,p]$.
\end{defi}
Note that for any $q,r \in \RR^2 \setminus \sigma$ with $q,r$ and $p$
pairwise distinct the arc $\arc_\sigma[q,r,p]$ is unique if it exists,
as $q,r,p$ defines a unique circle and there is at most
one starting point on $\sigma$ such that the arc starts to the left.
\begin{rmk}
The closure $\Gamma(\sigma, p)$ of $\Gamma^*(\sigma, p)$ is built
with respect to the parametric distance
$d_\infty : (\gamma_1, \gamma_2) \mapsto \sup_{t \in [0,1]} ||\gamma_1(t) - \gamma_2(t)||_2$.
Taking the closure is important as otherwise we would not necessarily have a maximal or minimal
element of $\Gamma(\sigma, p)$. Note that a connecting arc can also start in
$\sigma(0)$ or $\sigma(1)$ and in the case of $p$ being right of $\sigma$
a connecting arc may intersect $\sigma$ again at $\sigma(0)$ or $\sigma(1)$.
To be more precise: if $p$ is strictly left of $\sigma$, then
$\Gamma(\sigma, p) \setminus \Gamma^*(\sigma, p)$ is the set
\[
\{ \arc \in \Gamma \fdg \arc(0) \in \{\sigma(0), \sigma(1)\} \text{ and } \sprod{\normal{\sigma}{t(\arc(0))}}{\arc'(0)} \geq 0\}.
\]
If $p$ is strictly right of $\sigma$ then
$\Gamma(\sigma, p) \setminus \Gamma^*(\sigma, p)$ is the union of the two disjoint sets
\[
\{ \arc \in \Gamma \,:\, \arc(0) \in \{\sigma(0), \sigma(1)\} \text{ and }
\sprod{\normal{\sigma}{t(\arc(0))}}{\arc'(0)} \geq 0 \text{ and } \sigma \cap \gamma = \{\gamma(0)\}\}
\]
and
\[
\{ \arc \in \Gamma \fdg \arc(0) \in \sigma \text{ and } \gamma^\circ \cap \{\sigma(0), \sigma(1)\} \not= \emptyset \}.
\]
\end{rmk}
\begin{rmk}\label{remark_special_cuts_connecting_arcs}
Since with the definition from above one arc cannot cut another one
at its starting point or endpoint,
we use in Definition~\ref{defi_order_endpoint_fix},
for continuity reasons, the following extension of directional cuts.
Let $\arceins, \arczwei \in \Gamma(\sigma, p)$,
especially we have $\arceins(0), \arczwei(0) \in \sigma$,
and let $\arceins(0) \prec_\sigma \arczwei(0)$.
We say that $\arceins$ cuts $\arczwei$ in $1$ from the left and
$\arczwei$ cuts $\arceins$ in $1$ from the right if
$\arceins'(1) = \arczwei'(1)$ and $\gamma_2$ approaches $[\gamma_1]$ in $1$ from
the right, cf.~Figure~\ref{figure_extension_cut} (left).
Note that in this case every $\arcdrei \in \Gamma(\sigma, p)$ with
$\arceins(0) \prec_\sigma \arcdrei(0) \prec_\sigma \arczwei(0)$ and
$\arceins \cap \arcdrei = \arczwei \cap \arcdrei = \{p\}$ satisfies
$\arcdrei'(1) = \arceins'(1)$.
If $q := \arceins(0) \in \arczwei$, then we say $\arceins \capl \arczwei = \{0\}$ and
$\arczwei \capr \arceins = \{t_2(q)\}$, cf.~Figure~\ref{figure_extension_cut} (right).
Likewise, if $q := \arczwei(0) \in \arceins$ and $\arceins(0) \not\in \arczwei$, then we say
$\arceins \capl \arczwei = \{t_1(q)\}$ and $\arczwei \capr \arceins = \{0\}$.
The last two cases are only possible if $p$ is strictly right of $\sigma$.
\end{rmk}
\begin{figure}[ht]
\centering
\begin{overpic}[height=7cm]{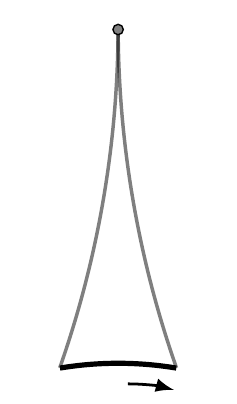}
\put(24,95){$p$}
\put(20,4){$\sigma$}
\put(14,30){$\gamma_1$}
\put(39,30){$\gamma_2$}
\end{overpic}
\hspace*{15mm}
\begin{overpic}[height=7cm]{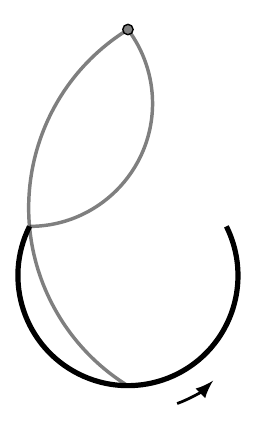}
\put(24,95){$p$}
\put(48,22){$\sigma$}
\put(34,57){$\gamma_1$}
\put(2,63){$\gamma_2$}
\end{overpic}
\caption{Illustration of extensions of directional cuts made in
Remark~\ref{remark_special_cuts_connecting_arcs}.\label{figure_extension_cut}}
\end{figure}
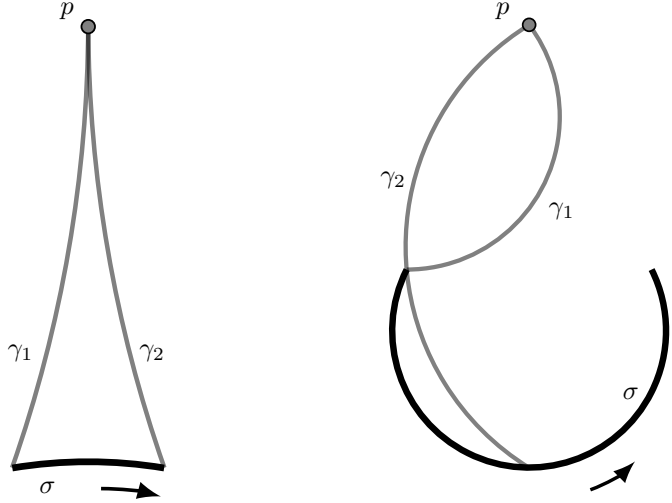
\begin{defi}\label{defi_order_endpoint_fix}
Let $\sigma \in \Gamma$, $p \in \RR^2 \setminus \sigma$ and $\arceins, \arczwei \in \Gamma(\sigma,p)$.
We say \emph{$\arceins \leq \arczwei$} if one of the following conditions holds:
\begin{enumerate}
\item $\arceins(0) \prec_\sigma \arczwei(0)$ and $\arceins \capl \arczwei = \emptyset$,
\item $\arceins(0) \succ_\sigma \arczwei(0)$ and $\arczwei \capl \arceins \not= \emptyset$,
\item $\arceins(0) = \arczwei(0) =: \sigma(t^*)$ and
$\sprod{\sigma'(t^*)}{\arceins'(0)} \leq \sprod{\sigma'(t^*)}{\arczwei'(0)}$.
\end{enumerate}
\end{defi}
An illustration of the three different cases can be found in Figure \ref{figure_examples_order}.
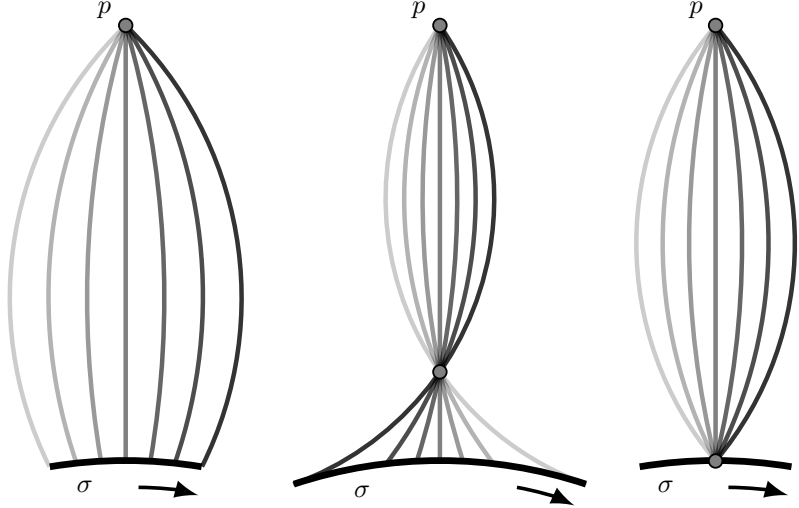
\begin{figure}[ht]
\centering
\begin{overpic}[height=7cm]{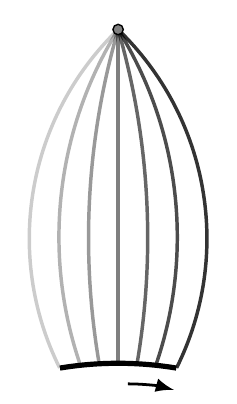}
\put(24,95){$p$}
\put(20,4){$\sigma$}
\end{overpic}
\begin{overpic}[height=7cm]{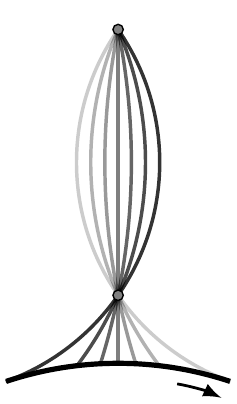}
\put(24,95){$p$}
\put(13,3){$\sigma$}
\end{overpic}
\begin{overpic}[height=7cm]{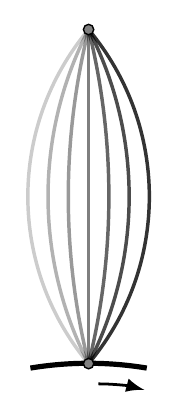}
\put(17,95){$p$}
\put(11,4){$\sigma$}
\end{overpic}
\caption{Illustration of the cases 1. (left), 2. (middle) and 3. (right)
of Definition \ref{defi_order_endpoint_fix}.
An arc is printed the darker the greater it is.\label{figure_examples_order}}
\end{figure}
\begin{rmk}
Note that in case 1 of Definition~\ref{defi_order_endpoint_fix},
which is $\gamma_1, \gamma_2 \in \Gamma(\sigma, p)$ with $\gamma_1(0) \prec_\sigma \gamma_2(0)$
and $\arceins \capl \arczwei = \emptyset$, the set $\arceins \capr \arczwei$ needs not
to be empty. Two examples are shown in Figure~\ref{figure_examples_order_2}.
\end{rmk}
\begin{figure}[ht]
\centering
\begin{overpic}[height=4cm]{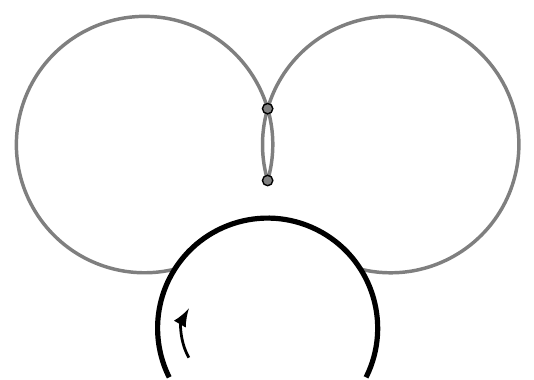}
\put(43,34){$p$}
\put(52,49.5){$x$}
\put(64,8){$\sigma$}
\put(7,50){$\arc_1$}
\put(88,50){$\arc_2$}
\end{overpic}
\hspace*{3mm}
\begin{overpic}[height=4cm]{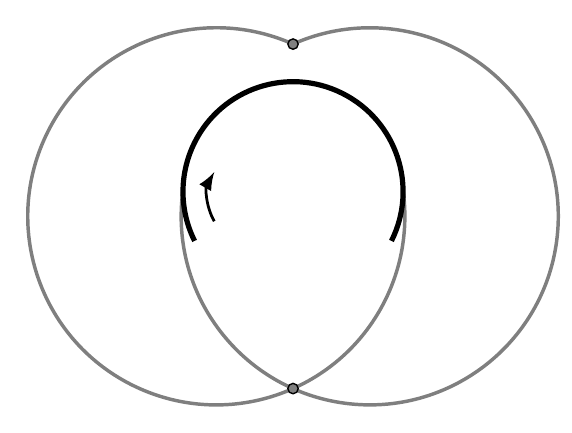}
\put(50,73){$p$}
\put(50,11){$x$}
\put(61,48){$\sigma$}
\put(26,25){$\arc_1$}
\put(71,25){$\arc_2$}
\end{overpic}
\caption{Two examples of arcs $\gamma_1, \gamma_2 \in \Gamma(\sigma, p)$ with
$\arceins(0) \prec_\sigma \arczwei(0)$ and $\arceins \capl \arczwei = \emptyset$
but having an intersection $x \not= p$, cf. case 1. of Definition \ref{defi_order_endpoint_fix}.
\label{figure_examples_order_2}}
\end{figure}
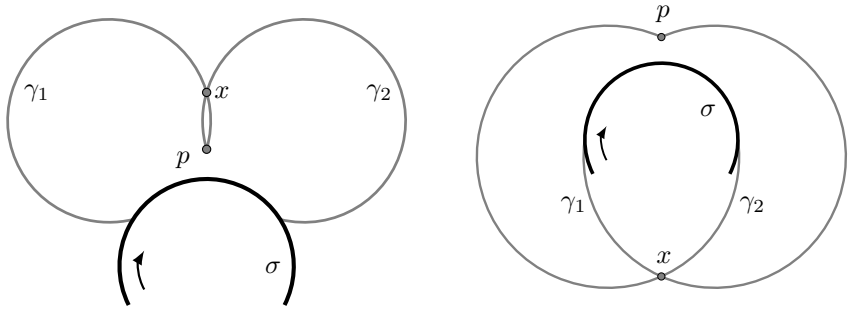
\begin{thm}
Let $\sigma$ be an arc and $p \in \RR^2 \setminus \sigma$. Then \enquote{$\leq$} is a total order 
on $\Gamma(\sigma, p)$.
\end{thm}
\begin{proof}
Reflexivity, totality and antisymmetry are immediate.
In order to prove the transitivity of \enquote{$\leq$}, 
let $\arceins, \arczwei, \arcdrei \in \Gamma(\sigma, p)$ be pairwise distinct with
$\arceins \leq \arczwei$ and $\arczwei \leq \arcdrei$.
We distinguish a total of six different cases based on the possible order 
of the starting points of $\arceins, \arczwei$ and $\arcdrei$ on $\sigma$.
We denote by $t_1,t_2,t_3 \in [0,1]$ the parameters such that
$\sigma(t_1) = \arceins(0), \sigma(t_2) = \arczwei(0)$ and $\sigma(t_3) = \arcdrei(0)$.\\
First, let $t_1 \leq t_2 \leq t_3$.
If $t_1 = t_3$ then $t_1 = t_2 = t_3$ and we get
$\sprod{\sigma'(t_1)}{\arceins'(0)} \leq \sprod{\sigma'(t_1)}{\arczwei'(0)} \leq \sprod{\sigma'(t_1)}{\arcdrei'(0)}$
which yields $\gamma_1 \leq \gamma_3$.\\
So, let $t_1 \not= t_3$ and assume the contrary, $\arcdrei < \arceins$.
By definition, $\arceins \capl \arcdrei \not= \emptyset$, so define
$\arceins \capl \arcdrei =: \{t_{13}\}$ and $\arcdrei \capr \arceins =: \{t_{31}\}$.
We consider the simple closed path
\[
\alpha := \sigma\restr{[t_1,t_3]} \sqcup \arcdrei\restr{[0,t_{31}]} \sqcup \reverse{\arceins\restr{[0,t_{13}]}}
\]
and let $Z$ be the unique connected component in $\RR^2 \setminus \alpha$
that is locally left of $\sigma\restr{[t_1,t_3]}$.
Then, $Z$ is locally right of $\arceins\restr{[0,t_{13}]}$ and
locally left of $\arcdrei\restr{[0,t_{31}]}$ and $\arczwei$ starts into $Z$
meaning that there is an $\varepsilon > 0$ so that
$\arczwei(t) \in Z$ for every $t \in (0,\varepsilon)$.
As $\arcdrei$ cuts $\arceins$ in $t_{31}$ from the right we know that if $t_{31} \not= 1$
then the endpoint $p$ of $\arcdrei$ is not in the closure of $Z$.
Hence, $\arczwei$ either has to cut $\arceins$ from the right or $\arcdrei$ from the left.
If $t_{31} = 1$, with the extensions of cuts for connecting arcs,
cf.~Remark~\ref{remark_special_cuts_connecting_arcs}, we get analogously that
$\arczwei$ cuts $\arceins$ from the right or $\arcdrei$ from the left.
Either way, this yields $\arczwei < \arceins$ or $\arcdrei < \arczwei$ which is a contradiction.\\
We now turn to the case $t_3 \leq t_2 \leq t_1$, $t_1 \not= t_3$.
If $t_1 = t_2$ we define $t_{12} := t_{21} := 0$, otherwise we know
that $\gamma_2 \capl \gamma_1 \not= \emptyset$ and we define
$\gamma_2 \capl \gamma_1 =: \{t_{21}\}$ and $\gamma_1 \capr \gamma_2 =: \{t_{12}\}$.
Likewise, we define $t_{23} := t_{32} := 0$ if $t_2 = t_3$,
$\gamma_3 \capl \gamma_2 =: \{t_{32}\}$ and $\gamma_2 \capr \gamma_3 =: \{t_{23}\}$, otherwise.
We assume $t_{21} \leq t_{23}$, otherwise the proof works analogously.
Similarly to the first case, consider the simple closed path
\[
\alpha := \sigma\restr{[t_3,t_2]} \sqcup \arczwei\restr{[0,t_{23}]} \sqcup \reverse{\arcdrei\restr{[0,t_{32}]}}
\]
and let $Z$ be the connected component in $\RR^2 \setminus \alpha$ that
is locally left of $\sigma\restr{[t_3,t_2]}$.
The case $t_{21} = t_{23} = 1$ is easily verified.
In the case $t_{21} = t_{23} \not= 1$ we know that $\gamma_1$ cuts $\gamma_3$ in $t_{12}$
since $\arc_1 \not= \arc_3$.
If $\gamma_1$ cuts $\gamma_3$ in $t_{12}$ from the left then $\gamma_1$ starts in
$t_{12}$ into $Z$. This would yields at least three intersections of $\gamma_1$ with either
$\gamma_2$ or $\gamma_3$ which would be a contradiction since then the arcs would coincide.
Hence, $\gamma_1$ intersects
$\gamma_3$ in $t_{12}$ from the right and we have $\gamma_1 \leq \gamma_3$.
If $t_{21} < t_{23}$ then $\gamma_1$ starts in $t_{12}$ into $Z$ as $Z$ is locally left
of $\gamma_2\restr{[0,t_{23}]}$ and $\gamma_1$ cuts $\gamma_2\restr{[0,t_{23}]}$ in $t_{12}$ from the right.
As in the first case, we can conclude that $\gamma_1$ has to cut
$\gamma_3\restr{[0,t_{32}]}$ from the right or $\gamma_2\restr{[0,t_{23}]}$ from the left.
As $\gamma_1$ and $\gamma_2$ cannot have three intersections,
$\gamma_1$ has to cut $\gamma_3$ from the right and we have $\gamma_1 \leq \gamma_3$.\\
Now, let us consider the remaining four possibilities how to order $t_1, t_2, t_3$.
We prove this cases by contradiction, so assume $\arcdrei < \arceins$.
In the case $t_1 \leq t_3 \leq t_2$, with $\arczwei \leq \arcdrei$, we are
in the setting of the second case and this yields $\arczwei < \arceins$,
a contradiction.
If $t_2 \leq t_1 \leq t_3$ then with $\arceins \leq \arczwei$ we also are
in the setting of the second case and get $\arcdrei < \arczwei$.
If $t_2 \leq t_3 \leq t_1$ or $t_3 \leq t_1 \leq t_2$ then with
$\arczwei \leq \arcdrei$ or $\arceins \leq \arczwei$ we are in
the setting or the first case and get $\arczwei < \arceins$
or $\arcdrei < \arczwei$, respectively.
Hence, in either case we get a contradiction.\qed
\end{proof}
For the rest of the chapter we treat the problem
of computing a maximal connecting arc with respect to \enquote{$\leq$}
from Definition \ref{defi_order_endpoint_fix}. Note that the minimal
connecting arc must not be considered as the respective results hold because of symmetry.
We use this in the initialization of Algorithm \ref{algo_klopfalgo}
in Section \ref{section_algorithm}.
Assume that $p_1, p_2, p_3$ are three pairwise distinct points.
We know that there exists a unique arc with starting point $p_1$,
endpoint $p_3$ that passes $p_2$ unless 
$p_1 \in [p_2,p_3]$ or $p_3 \in [p_1,p_2]$.
Because of Remark~\ref{remark_arcs_exist} we can ignore these cases
and we will assume that every such arc exists.
Likewise, we assume that for any two points $p_1, p_2, p_1 \not= p_2$ and any direction $\tau \in S^1$
the arc with starting point $p_1$, $\tau$ as unit tangent vector at the start and
endpoint $p_2$ exists.
\begin{prop}\label{prop_maximal_connecting_arc}
Let $\sigma$ be an arc and $p \in \RR^2 \setminus \sigma$.
If $p$ is strictly left of $\sigma$ or $p\in[\sigma]$ then
$\arc[\sigma'(1),\sigma(1),p]$
is the maximal connecting arc in $\Gamma(\sigma, p)$
with respect to \enquote{$\leq$}.
If $p$ is strictly right of $\sigma$ then the maximal connecting arc in $\Gamma(\sigma, p)$
with respect to \enquote{$\leq$} is
$\arc[\sigma(0),\sigma(1),p]$.
\end{prop}
\begin{proof}
First, let $p$ be strictly left of $\sigma$ and
let $\arceins = \arc[\sigma'(1),\sigma(1),p]$. Then,
$\arceins$ is in the boundary of $\Gamma^*(\sigma, p)$ and
by Definition \ref{defi_connecting_arc},
$\arceins$ is a connecting arc.
To show that $\arceins$ is maximal we take any $\arczwei \in \Gamma(\sigma, p)$
and show $\arczwei \leq \arceins$.
If $\arczwei(0) = \sigma(1)$ then
$\sprod{\sigma'(1)}{\arczwei'(0)} \leq 1 = \sprod{\sigma'(1)}{\arceins'(0)}$
which proves $\arczwei \leq \arceins$.\\
Now, let $\arczwei(0) \not= \sigma(1)$ and assume the contrary, $\arczwei > \arceins$.
As $\arczwei(0) \prec_\sigma \sigma(1) = \arceins(0)$, $\arczwei > \arceins$ yields
that $\arczwei$ cuts $\arceins$ from the left.
As $p$ is strictly left of $\sigma$ and $\arceins'(0) = \sigma'(1)$,
we know that $\sigma$ is right of $\arceins$.
Hence, $\arczwei$ starts strictly right of $\arceins$ which yields that
$\arczwei$ must cut $[\arceins]$ from the right before
$\arczwei$ can cut $\arceins$ from the left. This is a contradiction
to $\arceins \not= \arczwei$ as it would yield at least three intersections.\\
Now, let $p$ be strictly right of $\sigma$ and let $\arceins = \arc[\sigma(0),\sigma(1),p]$.
Again, $\arceins$ is a connecting arc as it is in the boundary of $\Gamma^*(\sigma,p)$.
Denote by $\alpha$ the simple closed path
$\arceins\restr{[0,t_1(\sigma(1))]} \sqcup \reverse{\sigma}$
and let $Z$ be the connected component of
$\RR^2 \setminus \alpha$ that is locally left of $\sigma$.
Let $\arczwei \in \Gamma(\sigma, p)$.
If $\arczwei(0) = \sigma(0)$, then we know that it does not start into $Z$
as otherwise $\arczwei$ would have to cut $\sigma^\circ$ or $\arc_1$ three times, as $p \not\in Z$.
Hence, if $\arczwei(0) = \sigma(0)$ then
$\sprod{\sigma'(0)}{\arczwei'(0)} \leq \sprod{\sigma'(0)}{\arceins'(0)}$.
If $\arczwei(0) \in \sigma^\circ$ then $\arczwei$ starts into $Z$ and as $p \not\in Z$
and $Z$ is locally right of $\arceins\restr{[0,t_1(\sigma(1))]}$ we know that
$\arczwei$ cuts $\arceins$ from the right, which yields $\arczwei \leq \arceins$.
If $\arc_2(0) = \sigma(1)$, then by Remark~\ref{remark_special_cuts_connecting_arcs}
$\arczwei$ cuts $\arceins$ in $0$ from the right and therefore $\arczwei \leq \arceins$.\\
Now, let $p\in[\sigma]$, $\arceins = \arc[\sigma'(1),\sigma(1),p]$ and $\arczwei \in \Gamma(\sigma, p)$.
If $\arczwei(0) = \sigma(1)$ then $\sprod{\sigma'(1)}{\arczwei'(0)} \leq 1 = \sprod{\sigma'(1)}{\arceins'(0)}$,
so we get $\arczwei \leq \arceins$.
If $\arczwei = \arc[-\sigma'(0),\sigma(0),p]$ then $\arceins \cap \arczwei = \{p\}$
and $\arceins'(1) = -\arczwei'(1)$, which yields $\arczwei \capl \arceins = \emptyset$, so $\arczwei < \arceins$.
Otherwise, we know that $\sprod{\normal{\sigma}{t_\sigma(\arczwei(0))}}{\arczwei'(0)} > 0$,
so $[\arczwei] \not= [\sigma] = [\arceins]$. As $\arczwei(0), \arczwei(1) \in [\arceins]$,
there cannot be a third point that is both on $\arczwei$ and on $[\arceins]$
and we know that $\arceins'(1) \not= \arczwei'(1)$.
This yields $\arczwei \capl \arceins = \emptyset$ and we get $\arczwei < \arceins$.\qed
\end{proof}
\begin{rmk}\label{rmk_maximal_connecting_arc}
Similarly, one can show that an arc $\arceins$ with $\arceins'(0) = \sigma'(t_\sigma(\arceins(0)))$
is greater than any arc $\arczwei$ with $\arczwei(0) \preceq_\sigma \arceins(0)$.
\end{rmk}
\begin{rmk}\label{remark_arcs_exist}
If $p$ is left of $\sigma$, then a maximal connecting arc obviously exists unless
$\sigma'(1)$ and the normalized vector $\sigma(1) - p$ are equal.
If $p$ is strictly right of $\sigma$, then a maximal connecting arc exists unless
$p$ is on the line segment $[\sigma(0),\sigma(1)]$ or $\sigma(0)$
is on the line segment $[p,\sigma(1)]$.
We know that the channel has finite diameter and so
the length of any arc $\arc\subset P$ is bounded from above.
This allows us to ignore arcs exceeding a certain maximal length
and to replace them by a respective arc having this maximal length.
In case of a maximal connecting arc, the maximal connecting arc with
restricted length will still be large enough for all requirements.
Hence, we will simply assume that the maximal connecting arc always exists.
\end{rmk}
\section{Restrictions\label{section_restrictions}}
In this chapter, $\sigma$ will denote the starting arc and
$\kappa$ the boundary of the channel $P$ as stated in
Definition~\ref{defi_channel} and we have $p \in P^\circ$.\\
Restrictions will be the key tool to characterize circular visibility.
Most papers on circular visibility use
the well-known fact that the boundary of the set of circularly visible
points from a given point or edge consists of arcs having two,
or, in the case of the edge, three points in common with the channel boundary,
cf.~\cite{agarwal1993, chou1992, chou1995, Maier2014}.
An example is shown in Figure~\ref{figure_channel_visibility_set},
in which the boundary arc of the visibility set is shown dashed.
We will use a similar approach, but we extend the definition of
restrictions to arcs that are not completely inside the channel.
An example is shown in Figure~\ref{figure_alternating_sequence} (right):
we have a connecting arc $\arc$ with endpoint $p \in P^\circ$ but it is not
a visibility arc as $\arc \not\subset P$.
The question we are interested in is if we can modify $\arc$
such that we obtain a visibility arc ending in $p$ or not.\\
The point $r_2$ shows an example of a restriction from the left
in the \enquote*{classical} sense: $\arc$ would immediately leave the channel
if we would move it slightly to the left at $r_2$.\\
The situation at the point $r_1$ is slightly different:
$\arc$ is leaving the channel in this region.
Nevertheless, since $\arc$ would run within the channel if we would
push it to the left there, we will call this also a \textit{restriction}; 
certainly \enquote*{from the right} in this case, however.
The part of the channel, where the arc is running outside the channel,
shown as thick line, will be called a \textit{violation}.
\begin{rmk}\label{remark_special_cuts}
For continuity reasons, we extend the definition of directional cuts
for the channel boundary. Let $\arc \in \Gamma(\sigma, p)$.
We say that $\kappa$ approaches $\arc$ in $0$ from the right if
$\kappa(0) \in \arc$. If $\arc(0) = \kappa(0)$ then we say that
$\kappa$ leaves $\arc$ in $0$ to the left or right if it leaves
$[\arc]$ in $0$ to the left or right, respectively.
Similarly, if $\arc(0) = \kappa(1)$ then we say that $\kappa$
approaches $\arc$ in $1$ from the left or right if it approaches
$[\arc]$ in $1$ from the left or right, respectively.
\end{rmk}
\begin{figure}[ht]
\hspace*{-20mm}
\begin{overpic}[height=5cm,angle=-33]{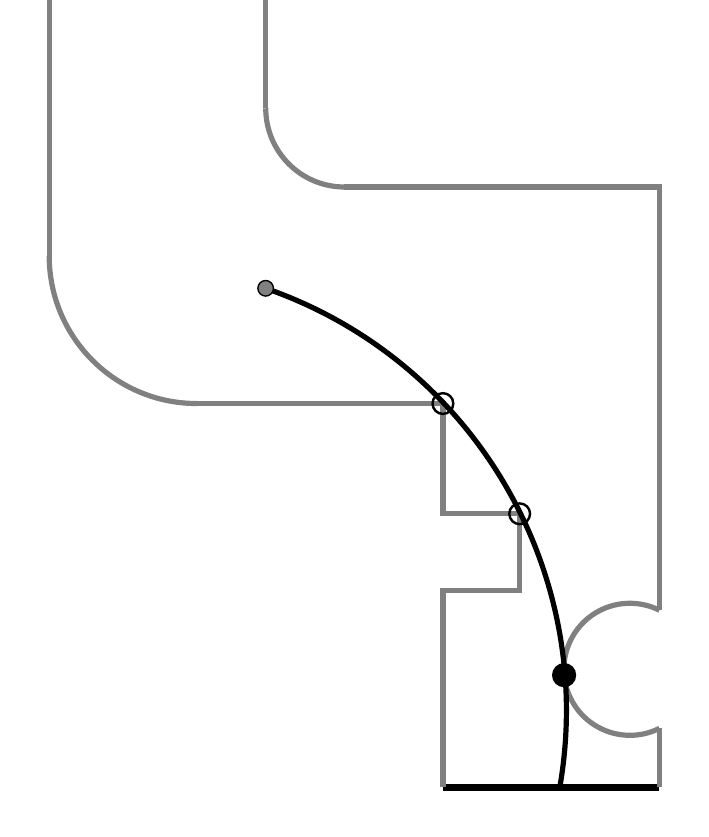}
\put(43,67){$p$}
\put(40,4){$\sigma$}
\put(54,16){$r_1$}
\put(58,33){$r_2$}
\put(58,45){$r'_2$}
\end{overpic}
\hspace*{-25mm}
\begin{overpic}[height=5cm,angle=-33]{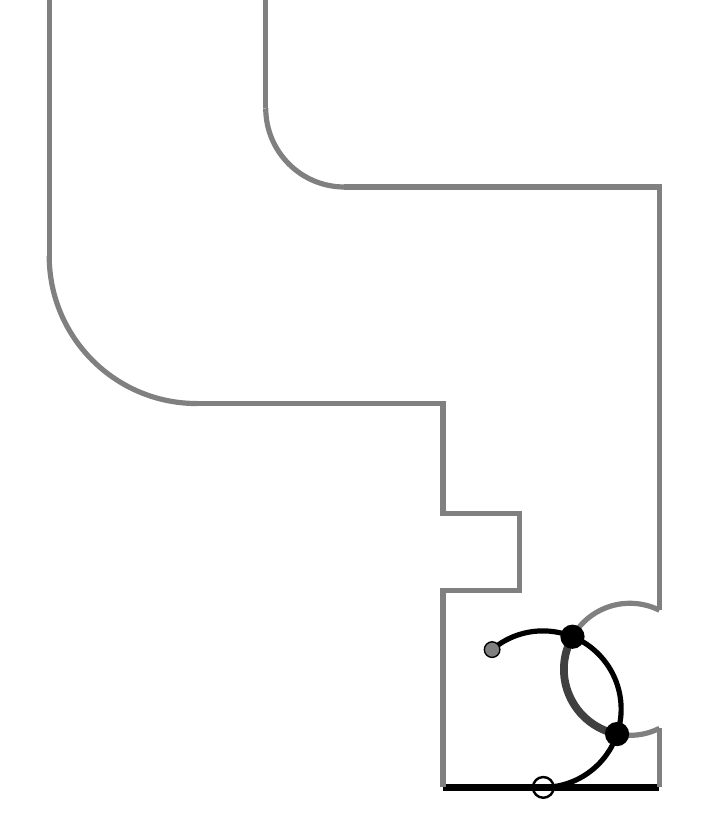}
\put(43,21){$p$}
\put(44,4){$\sigma$}
\put(28,8){$r_1$}
\put(35,10){\vector(4,1){8}}%
\put(64,14){$r_2$}
\put(63,15){\vector(-4,-1){8}}%
\put(58,34){$r'_2$}
\put(59,32){\vector(-1,-2){4}}%
\end{overpic}
\hspace*{-25mm}
\begin{overpic}[height=5cm,angle=-33]{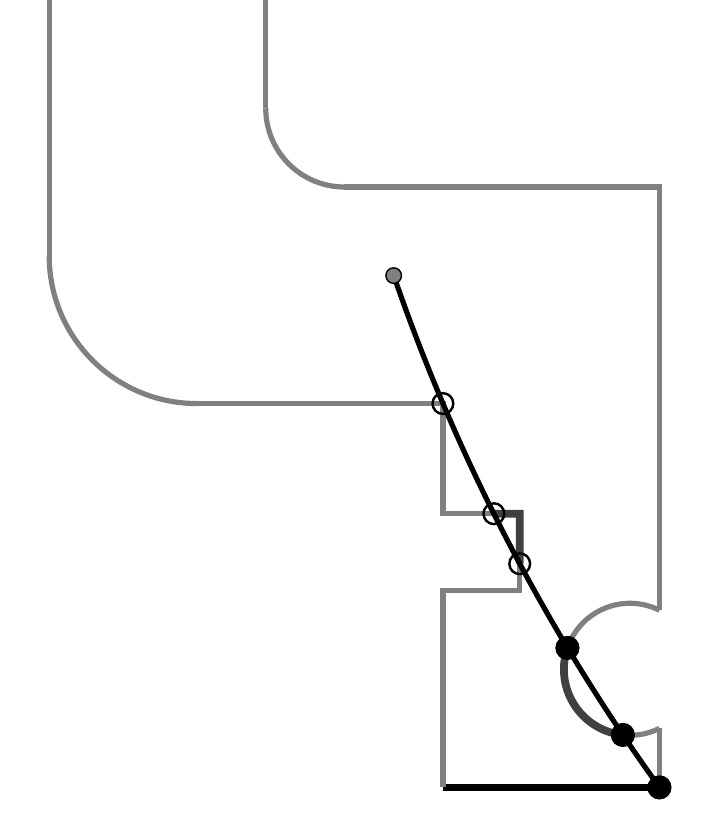}
\put(53,61){$p$}
\put(40,4){$\sigma$}
\end{overpic}\\
\caption{Illustration of restrictions and violations\label{figure_alternating_sequence}}
\end{figure}
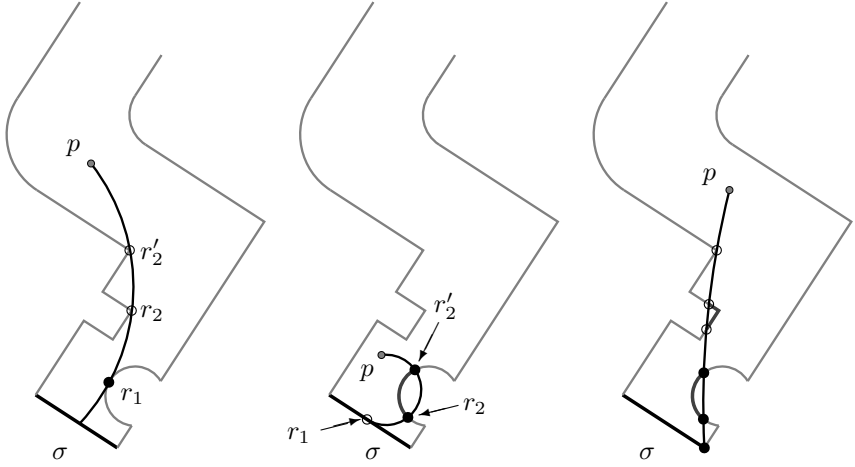
\begin{defi}[Restriction\label{defi_restriction}]
Let $\alpha : [0,1] \to \RR^2$ be an arc spline and $\arc$ an arc.
Then we define $\Delta(\gamma, \alpha) \in \ZZ$ by
\begin{align*}
 \Delta(\gamma, \alpha) := \;\,
   & |\{t \in [0,1] \fdg \alpha \text{ approaches } \arc \text{ in } t \text{ from the right} \} |\\
 - \; & |\{t \in [0,1] \fdg \alpha \text{ approaches } \arc \text{ in } t \text{ from the left} \} | \\
 + \; & |\{t \in [0,1] \fdg \alpha \text{ leaves } \arc \text{ in } t \text{ to the left} \} |\\
 - \; & |\{t \in [0,1] \fdg \alpha \text{ leaves } \arc \text{ in } t \text{ to the right} \} |.
\end{align*}
With $p\in P^\circ$, $\arc\in\Gamma(\sigma, p)$ and $q \in \kappa$, we define
\[
\Delta_q(\gamma, \kappa) := \Delta(\gamma, \kappa\restr{[0, t_\kappa(q)]}).
\]
We call $q \in \kappa$ a \emph{restriction from the right of $\gamma$} if
$\Delta_q(\gamma, \kappa) = 1$, a \emph{restriction from the left of $\gamma$} if
$\Delta_q(\gamma, \kappa) = -1$, a \emph{violation from the right of $\gamma$} if
$\Delta_q(\gamma, \kappa) > 1$ and a \emph{violation from the left of $\gamma$} if
$\Delta_q(\gamma, \kappa) < -1$.\\
Furthermore, we call $q = \arc(0)$ a \emph{starting restriction from the left} if $\sigma$
is locally right of $[\arc]$ at $\arc(0)$ and we call it a 
\emph{starting restriction from the right} if $\sigma$
is locally left of $[\arc]$ at $\arc(0)$.
\end{defi}
\begin{rmk}
In Figure~\ref{figure_alternating_sequence} restrictions and violations are
depicted for some connecting arcs.
Restrictions from the left are depicted as small circles, restrictions from the
right are depicted as filled circles and the bold dark parts of the channel represent
violations. On the left-hand side there is one restriction
from the right, $r_1$, and two restrictions from the left, $r_2$ and $r'_2$,
and there is no violation.
The figure in the middle has a starting restriction from the left, $r_1$,
two restrictions from the right, $r_2$ and $r'_2$, and every point of the dark
part of the channel that is between $r_2$ and $r'_2$ is a violation from the right.
\end{rmk}
\begin{rmk}\label{rmk_violation_tolerance}
In practice, it is often neither necessary nor reasonable to regard every small violation.
In that case you simply regard only relevant violations having a distance to
the respective arc that is greater than a certain threshold.
This threshold is a parameter that can be chosen freely and is only dependent of
the application.
It defines how strictly the border of the channel has to be regarded.
\end{rmk}
If $\alpha$ is a simple closed arc spline and $\arc$ is an arc with
$\arc(0), \arc(1) \not\in \alpha$ then
the following lemma shows that the value of $\Delta(\arc, \alpha)$
can be determined just by analyzing in which connected
component of $\RR^2 \setminus \alpha$ the points $\arc(0)$ and $\arc(1)$
are. Note that this result is connected to the behavior of the
winding number along a path, known from algebraic topology or complex analysis,
see \cite{roe2015}.
\begin{prop}\label{prop_delta_vanishes}
Let $\alpha$ be a simple closed arc spline whose interior $I_\alpha$
is locally left of $\alpha$ and $\arc$ an arc with
$\arc(0), \arc(1) \not\in \alpha$.\\
If $\arc(0), \arc(1) \in I_\alpha$ or $\arc(0), \arc(1) \not\in I_\alpha$ then $\Delta(\arc, \alpha) = 0$.
If $\arc(0) \in I_\alpha$ and $\arc(1) \not\in I_\alpha$ then $\Delta(\arc, \alpha) = 2$ and
if $\arc(0) \not\in I_\alpha$ and $\arc(1) \in I_\alpha$ then $\Delta(\arc, \alpha) = -2$.
\end{prop}
\begin{proof}
Let $\beta_1, \beta_2, \ldots. \beta_n$ be arcs such that
$\beta := \arc \sqcup \beta_1 \sqcup \cdots \sqcup \beta_n$ is simple closed
and for $i \in \{1, \ldots, n\}$ we have $\beta_i \cap \alpha = \{p_i\}$,
$p_i \in \beta_i^\circ$, $p_i$ is not a breakpoint of $\alpha$
and $\alpha$ cuts $\beta_i$ in $t_\alpha(p_i)$ either from the left
or from the right.
To simplify notation we write $\beta_0 := \gamma$.\\
First, we show $\sum_{i=0}^n\Delta(\beta_i, \alpha) = 0$.
We can always assume that $\alpha(0) \not\in \beta$ as $\alpha$ is closed
and we define $0 = t_1 < t_2 < \cdots < t_m = 1$ such that
for every $j \in \{1, \ldots, m - 1\}$ we have $\alpha(t_j) \not\in \beta$
and $\{t \in [t_j, t_{j+1}] \fdg \alpha(t) \in \beta\}$
is a single nonempty interval.
As $\beta$ is simple closed, the interior $I_\beta$ and the exterior
$E_\beta$ of $\beta$ are well-defined and we know that either $I_\beta$
is locally left of $\beta_j$ and $E_\beta$ is locally right of $\beta_j$
for every $j\in\{0, \ldots, n\}$ or vice versa.
We assume that $I_\beta$ is locally left of $\beta_0$,
since the arguments work analogously, otherwise.
Let $j\in\{1, \ldots, m-1\}$.
If $\alpha(t_j) \in I_\beta$ and $\alpha(t_{j+1}) \in E_\beta$
then we know that there is exactly one $t \in (t_j, t_{j+1})$,
one $t' \in (t_j, t_{j+1})$ and one $i \in \{0, \ldots, n\}$ such that
$\alpha$ approaches $\beta_i$ in $t$ from the left and $\alpha$ leaves
$\beta_i$ in $t'$ to the right.
With this unique index $i$, we know that
$\Delta(\beta_k, \alpha\restr{[t_j, t_{j+1}]}) = -2$ if $k=i$
and that it vanishes, otherwise. This yields
$\sum_{i=0}^n\Delta(\beta_i, \alpha\restr{[t_j, t_{j+1}]}) = -2$.
Analogously, we get $\sum_{i=0}^n\Delta(\beta_i, \alpha\restr{[t_j, t_{j+1}]}) = 2$
if $\alpha(t_j) \in E_\beta$ and $\alpha(t_{j+1}) \in I_\beta$
and we get that $\sum_{i=0}^n\Delta(\beta_i, \alpha\restr{[t_j, t_{j+1}]})$
vanishes if both $\alpha(t_j)$ and $\alpha(t_{j+1})$ are either in
$I_\beta$ or in $E_\beta$.
As $\alpha$ is closed and $\alpha(t_m) = \alpha(t_1)$
there are just as many $j \in \{1, \ldots, m - 1\}$ with
$\sum_{i=0}^n\Delta(\beta_i, \alpha\restr{[t_j, t_{j+1}]}) = -2$
as with $\sum_{i=0}^n\Delta(\beta_i, \alpha\restr{[t_j, t_{j+1}]}) = 2$.
This yields $\sum_{i=0}^n\Delta(\beta_i, \alpha) =
\sum_{j=1}^{m-1}\sum_{i=0}^n\Delta(\beta_i, \alpha\restr{[t_j, t_{j+1}]}) = 0$.\\
Now, let $i \in \{1, \ldots, n\}$, $\beta_i(0) \in I_\alpha$
and let $\alpha = \alpha_1 \sqcup \alpha_2 \sqcup \cdots \sqcup \alpha_m$
with arcs $\alpha_1, \ldots, \alpha_m, m\in \NN$.
We denote the exterior of $\alpha$ by $E_\alpha$.
Since $\alpha$ cuts $\beta_i$ exactly once from the left or from the right
and since no intersection is a breakpoint of $\beta$ or $\alpha$, we know that
$\beta_i(1) \in E_\alpha$ and that there is a unique $j \in \{1, \ldots, m\}$
such that $\beta_i \cap \alpha_j \not= \emptyset$.
As $I_\alpha$ is locally left of $\alpha_j$, we know that $\beta_i$
cuts $\alpha_j$ from the left. Hence, there is exactly one $t\in[0,1]$
such that $\alpha_j(t) \in \beta_i$ and $\alpha_j$ cuts $\beta_i$ in $t$ from the right.
This yields $\Delta(\beta_i, \alpha) = \Delta(\beta_i, \alpha_j) = 2$.
Analogously, we get $\Delta(\beta_i, \alpha) = -2$ if $\beta_i(0) \in E_\alpha$.\\
Assume $\arc(0) \in I_\alpha$ and $\arc(1) \in E_\alpha$. Then,
$\beta_1(0) \in E_\alpha, \beta_n(1) \in I_\alpha$ and
$\beta_{i+1}(0) = \beta_i(1) \in I_\alpha$ if $\beta_{i}(0) \in E_\alpha, i\in\{1, \ldots, n-1\}$
and vice versa.
Hence, $\Delta(\arc, \alpha) = \Delta(\beta_0, \alpha)
= - \sum_{i=1}^n \Delta(\beta_i, \alpha) = - (-2 + 2 - \cdots + 2 - 2) = 2$.
Similarly, one can prove the remaining three cases.\qed
\end{proof}
\begin{lem}\label{lemma_restriction_well_defined}
Let $p \in P^\circ$ and $\arc \in \Gamma(\sigma, p)$.
We have $\Delta(\arc, \kappa) = 0$ if $\kappa(1) \not\in \arc$
and $\Delta(\arc, \kappa) = -1$, otherwise.
\end{lem}
\begin{proof}
First, assume $\arc(0) \in \sigma^\circ$.
Then there exists an $\varepsilon > 0$ such that
${\arc([0,\varepsilon]) \cap \kappa} = \emptyset$ and we know that $\arc(\varepsilon)$
and $\arc(1)$ are in $P^\circ$, that is in the same connected component of $\RR^2 \setminus (\sigma \sqcup \kappa)$.
With Proposition~\ref{prop_delta_vanishes} we get that ${\Delta(\arc\restr{[\varepsilon, 1]}, \sigma \sqcup \kappa)}$ vanishes.
If $\arc\restr{[\varepsilon, 1]} \cap \sigma = \emptyset$ then we get
$\Delta(\arc\restr{[\varepsilon, 1]}, \kappa) = 0$ and finally $\Delta(\arc, \kappa) = 0$.
Otherwise, either $\sigma$ approaches $\arc\restr{[\varepsilon, 1]}$ in $1$
from the right or it leaves $\arc\restr{[\varepsilon, 1]}$ in $0$ to the left.
In the first case this yields $\Delta(\arc\restr{[\varepsilon, 1]}, \kappa) = -1$
as $\Delta(\arc\restr{[\varepsilon, 1]}, \sigma) = 1$.
Taking into account Remark~\ref{remark_special_cuts},
we get $\Delta(\arc\restr{[\varepsilon, 1]}, \kappa) = 0$ as $\kappa$ is said to
approach $\gamma$ in $0$ from the right if $\kappa(0) \in \gamma$.
Hence, we get $\Delta(\arc, \kappa) = 0$ in this case too.
If $\sigma$ leaves $\arc\restr{[\varepsilon, 1]}$ in $0$ to the left, then also
$\Delta(\arc\restr{[\varepsilon, 1]}, \kappa) = -1$ as
$\Delta(\arc\restr{[\varepsilon, 1]}, \sigma) = 1$.
As in this case $\kappa(1) = \sigma(0) \in \gamma$, this is the claimed statement.\\
Now, assume $\arc(0) = \sigma(1)$.
If $\kappa$ leaves $[\arc]$ in $0$ to the right then there exists an $\varepsilon > 0$
such that $\arc((0,\varepsilon)) \cap \kappa = \emptyset$ and $\arc(\varepsilon) \in P^\circ$.
As in the first case, we get $\Delta(\arc\restr{[\varepsilon, 1]}, \kappa) = 0$.
Again by Remark~\ref{remark_special_cuts}, $\kappa$ approaches $\arc([0,\varepsilon])$
in $0$ from the right and it leaves $\arc([0,\varepsilon])$ in $0$ to the right.
This yields $\Delta(\arc\restr{[0, \varepsilon]}, \kappa) = 0$, thus $\Delta(\arc, \kappa) = 0$.
If $\kappa$ leaves $[\arc]$ in $0$ to the left then there exists an $\varepsilon > 0$
such that $\arc((0,\varepsilon)) \cap \kappa = \emptyset$ and $\arc(\varepsilon) \in \RR^2 \setminus P$.
As Proposition~\ref{prop_delta_vanishes} yields $\Delta(\arc\restr{[\varepsilon, 1]}, \sigma \sqcup \kappa) = -2$,
we get $\Delta(\arc\restr{[\varepsilon, 1]}, \kappa) = -2$, analogously to the first case.
Again by Remark~\ref{remark_special_cuts} we get
$\Delta(\arc\restr{[0, \varepsilon]}, \kappa) = 2$ as $\kappa$ approaches
$\arc\restr{[0, \varepsilon]}$ in $0$ from the right and it leaves
$\arc\restr{[0, \varepsilon]}$ in $0$ to the left, which yields $\Delta(\arc, \kappa) = 0$.
If $\kappa$ neither leaves $\arc$ in $0$ to the left nor to the right then
we can use the same argumentation for the first $t \in [0,1]$ in which
$\kappa$ leaves $\arc$. The case $\arc(0) = \sigma(0)$ can be shown analogously.\qed
\end{proof}
\begin{rmk}
$\Delta_q(\gamma, \kappa)$ is odd if and only if $q \in \gamma$, so
a restriction of $\gamma$ is on $\gamma$.
\end{rmk}
\begin{defi}[Alternating Sequence]
We call $(a_1, a_2, \ldots, a_n), n\in\NN$, an \emph{alternating sequence of length $n$}
of a connecting arc $\arc$ if $a_1 \preceq_{\arc} a_2 \prec_{\arc} \cdots \prec_{\arc} a_n$,
if $a_1, a_3, \ldots$ are restrictions from the left and
$a_2, a_4, \ldots$ are restrictions from the right or vice versa.\\
We call $(a_1, a_2, \ldots, a_n)$ \emph{left-blocking} if $a_1$
is a restriction from the left and \emph{right-blocking}, otherwise.
\end{defi}
Note that $a_1$ must be a starting restriction if $a_1 = a_2$.
In Figure~\ref{figure_alternating_sequence} (left and right)
one can see a connecting arc having a right-blocking alternating sequence of length two,
in the middle there is a left-blocking alternating sequence of length two.
Note that there is no alternating sequence of length three in any of the three figures.\\
Good examples of alternating sequences are also depicted in Figure~\ref{figure_algorithm},
although the alternating sequences are not marked explicitly.
In every figure the connecting arc has a right-blocking alternating sequence of length
two.
\begin{lem}\label{lemma_delta_difference_on_arcs}
Let $p \in P^\circ$, $\arceins, \arczwei \in \Gamma(\sigma, p)$, $\arceins < \arczwei$ and
$q \in (\arceins \cup \arczwei) \cap \kappa$.
\begin{enumerate}
\item If $\arceins(0) \prec_\sigma \arczwei(0)$ and $\arceins \cap \arczwei \subset \{p\}$
then $\Delta_q(\arczwei, \kappa) - \Delta_q(\arceins, \kappa) = 1$.
\item If $\arceins(0) = \arczwei(0)$ then
\[
\Delta_q(\arczwei, \kappa) - \Delta_q(\arceins, \kappa) =
\begin{cases}
0, & \text{ if } q = \arceins(0),\\
1, & \text{ otherwise.}\hspace*{25mm}
\end{cases}
\]
\item If $\arceins(0) \prec_\sigma \arczwei(0)$ and $\arceins \cap \arczwei \not\subset \{p\}$
then with $\arceins \capr \arczwei =: \{t_1\}$ and $\arczwei \capl \arceins =: \{t_2\}$ we know:
\[
\Delta_q(\arczwei, \kappa) - \Delta_q(\arceins, \kappa) =
\begin{cases}
 2, & \text{ if } q = \arceins(t_1),\\
 1, & \text{ if } q \in \arceins\restr{[0,t_1)} \cup \arczwei\restr{[0,t_2)},\\
 3, & \text{ if } q \in \arceins\restr{(t_1,1)} \cup \arczwei\restr{(t_2, 1)}.\hspace*{3mm}
\end{cases}
\]
\item If $\arczwei(0) \prec_\sigma \arceins(0)$ then with
$\arceins \capr \arczwei =: \{t_1\}$ and $\arczwei \capl \arceins =: \{t_2\}$ we have
\[
\Delta_q(\arczwei, \kappa) - \Delta_q(\arceins, \kappa) =
\begin{cases}
 0, & \text{ if } q = \arceins(t_1),\\
-1, & \text{ if } q \in \arceins\restr{[0,t_1)} \cup \arczwei\restr{[0,t_2)},\\
 1, & \text{ if } q \in \arceins\restr{(t_1,1)} \cup \arczwei\restr{(t_2, 1)}.
\end{cases}
\]
\end{enumerate}
\end{lem}
\begin{proof}
We just consider the case $\arc_2(0) \prec_\sigma \arc_1(0)$
as the other cases can be proven analogously.
To simplify notation, let $S_1 := \arc_1\restr{[0,t_1)}$,
$S_2 := \arc_2\restr{[0,t_2)}$, $S_3 := \arc_1\restr{(t_1, 1)}$,
$S_4 := \arc_2\restr{(t_2, 1)}$, $S_5 := \{\arc_1(t_1)\}$
and $I := \{1, \ldots, 5\}$.\\
Let $j,k \in I$ and let $\alpha : [0,1] \to \RR^2$ be a simple arc spline
such that $\alpha(0) \in S_j$, $\alpha(1) \in S_k$ and
$\alpha \cap S_l = \emptyset$ for $l \in I \setminus \{j,k\}$ and such that
the set\linebreak $\{t \in [0,1] \fdg \alpha(t) \in S_j\}$ is a single interval
as well as $\{t \in [0,1] \fdg \alpha(t) \in S_k\}$.
Then $\Delta(\arc, \alpha)$ can be easily computed.
For example, suppose $\alpha(0) \in S_1$ and $\alpha(1) \in S_2$.
$\alpha$ must leave $\arc_1$ either to the left or to the right.
If $\alpha$ leaves $\arc_1$ to the left, then
$\alpha$ has to approach $\arc_2$ from the right, so
$\Delta(\arc_2, \alpha) - \Delta(\arc_1, \alpha) = 1 - 1 = 0$.
If $\alpha$ leaves $\arc_1$ to the right then
$\alpha$ has to approach $\arc_2$ from the left and we get
$\Delta(\arc_2, \alpha) - \Delta(\arc_1, \alpha) = (-1) - (-1) = 0$.
Hence, if $\alpha(0) \in S_1, \alpha(1) \in S_2$ then the value of
$\Delta(\arc_2, \alpha) - \Delta(\arc_1, \alpha)$ is independent of the
actual path.
It can be easily checked that this is true for every $j,k \in I$.
We show the values of $\Delta(\arc_2, \alpha) - \Delta(\arc_1, \alpha)$
with $\alpha(0) \in S_j, \alpha(1) \in S_k, j, k \in I$,
in the following denoted by $\Delta_{j,k}$, in Table~\ref{table_djk}.
\begin{table}[H]
\begin{center}
\begin{tabular}{r|ccccc}
  &  1 &  2 & 3 & 4 & 5\\ \hline
1 &  0 &  0 & 2 & 2 & 1\\
2 &  0 &  0 & 2 & 2 & 1\\
3 & -2 & -2 & 0 & 0 & -1\\
4 & -2 & -2 & 0 & 0 & -1\\
5 & -1 & -1 & 1 & 1 & 0
\end{tabular}
\end{center}
\caption{Values of $\Delta_{j,k}$ with $j$ the row and $k$ the column\label{table_djk}}
\end{table}
\noindent It is easy to verify that $\Delta_{j,k} + \Delta_{k,l} = \Delta_{j,l}$ holds for any $j,k,l \in I$,
which yields $\Delta_{j_1,j_2} + \Delta_{j_2,j_3} + \cdots + \Delta_{j_{n-1},j_n} = \Delta_{j_1, j_n}$,
$n \in \NN, j_1, \ldots, j_n \in I$.\\
Now, let us consider $\Delta_q(\arc_2, \kappa) - \Delta_q(\arc_1, \kappa)$ with $q \in S_k, k\in I$.
We know that $\Delta_q(\arc_2, \kappa) - \Delta_q(\arc_1, \kappa)
= \Delta(\arc_2, \kappa\restr{[0,t_\kappa(q)]}) - \Delta(\arc_1, \kappa\restr{[0,t_\kappa(q)]})$.
Let $t' \in [0,1]$ be the first parameter such that $\kappa(t') \in \arc_1 \cup \arc_2$,
and let $\kappa(t') \in S_j, j \in I$.
As $\kappa\restr{[t', t_\kappa(q)]}$ can be split in parts satisfying the condition
from above, we know that $\Delta_q(\arc_2, \kappa) - \Delta_q(\arc_1, \kappa) =
\Delta(\arc_2, \kappa\restr{[0,t']}) - \Delta(\arc_1, \kappa\restr{[0,t']}) + \Delta_{j,k}$.
Hence, the claim is easily verified for every case.\qed
\end{proof}
\begin{lem}\label{lemma_two_alternates_order}
Let $p \in P^\circ$, $\arceins, \arczwei \in \Gamma(\sigma, p)$ and let $(a_1, a_2)$ be
a left-blocking alternating sequence of $\arceins$.
If $\arceins < \arczwei$ then $a_1$ is a violation from the left of $\arczwei$
or $a_2$ is a violation from the right of $\arczwei$.
\end{lem}
\begin{proof}
If $a_1 = a_2$ then we know that $a_1$ is a starting restriction from the left
and $a_2 = \sigma(1)$ and this yields that $\arceins$ is the maximal arc
in $\Gamma(\sigma, p)$. As we assumed $\arceins < \arczwei$,
we know that $a_1 \not= a_2$.\\
If $\arczwei(0) \prec_\sigma \arceins(0)$ then, with
$\arceins \capr \arczwei =: \{t\}$, we know that
$t_1(a_1) < t$ or $t < t_1(a_2)$.
By Lemma~\ref{lemma_delta_difference_on_arcs}, we get
$\Delta_{a_1}(\arczwei, \kappa) - \Delta_{a_1}(\arceins, \kappa) = -1$
if $t_1(a_1) < t$ and 
$\Delta_{a_2}(\arczwei, \kappa) - \Delta_{a_2}(\arceins, \kappa) = 1$
if $t < t_1(a_2)$. This yields $\Delta_{a_1}(\arczwei, \kappa) = -2$
if $t_1(a_1) < t$ and $\Delta_{a_2}(\arczwei, \kappa) = 2$, otherwise.\\
If $\arceins(0) \preceq_\sigma \arczwei(0)$ then, by Lemma~\ref{lemma_delta_difference_on_arcs},
we know that $\Delta_{a_2}(\arczwei, \kappa) - \Delta_{a_2}(\arceins, \kappa) > 0$,
as $a_2 \in \arceins\restr{(0,1)}$. Hence, $\Delta_{a_2}(\arceins, \kappa) = 1$ yields
$\Delta_{a_2}(\arczwei, \kappa) > 1$.\qed
\end{proof}
Note that Lemma~\ref{lemma_two_alternates_order} holds analogously
if $(a_1, a_2)$ is a right-blocking alternating sequence and $\arczwei < \arceins$.
\begin{thm}\label{lemma_three_alternates_visible}
Let $p \in P^\circ$, $\arc \in \Gamma(\sigma,p)$ and let $(a_1, a_2, a_3)$ be an
alternating sequence of $\arc$. If $\arc \not\subset P$ then $p$ is not visible.
\end{thm}
\begin{proof}
Assume that there is a $\arczwei \in \Gamma(\sigma, p)$ with $\arczwei \subset P$.
As $\arceins$ and $\arczwei$ are not equal and as \enquote{$\leq$} is a total order
we have either $\arceins < \arczwei$ or $\arczwei < \arceins$.
We know that $(a_1, a_2)$ is a left-blocking alternating sequence
and $(a_2, a_3)$ is a right-blocking alternating sequence or vice versa.
In either case we can apply Lemma~\ref{lemma_two_alternates_order}
and get that there is a violation of $\arczwei$.\qed
\end{proof}
\section{Algorithm\label{section_algorithm}}
Recall that $\kappa = \kappa_1 \sqcup \kappa_2 \sqcup \cdots \sqcup \kappa_n$
is the channel boundary, an arc spline with $n \in \NN$ segments.
Let $p \in P^\circ$ and $\arc \in \Gamma(\sigma, p)$.
We say that a channel segment $\kappa_j, j\in\{1,\ldots, n\}$ is a \emph{restriction from
the left of $\arc$} if
$\{-1\} \subseteq \{{\Delta_q(\arc, \kappa) \fdg q \in \kappa_j}\} \subseteq \{-1, 0\}$.
Analogously, we say that $\kappa_j$ is a \emph{restriction from the right of
$\arc$} if $\{1\} \subseteq \{\Delta_q(\arc, \kappa) \fdg q \in \kappa_j\} \subseteq \{1, 0\}$.
We call $\kappa_j$ a \emph{violation from the left of $\arc$} if there is a $q \in \kappa_j$ with
$\Delta_q(\arc, \kappa) < -1$ and a \emph{violation from the right of $\arc$} if there is a
$q \in \kappa_j$ with $\Delta_q(\arc, \kappa) > 1$.
With $\kappa_j, \kappa_k, j,k\in\{1,\ldots,n\}$ being two restrictions of $\arc$,
we write $\kappa_j \prec_\arc \kappa_k$
if there are Restrictions $q_1 \in \kappa_j, q_2 \in \kappa_k$ with $q_1 \prec_\arc q_2$.
With this in mind, we can define an alternating sequence $(\kappa_j, \kappa_k)$.
Note that a channel segment cannot be both a restriction from the left
and a restriction from the right at the same time.\\
In this chapter we give an algorithm and prove that it computes in linear time
an arc $\arc \in \Gamma(\sigma, p)$ that is either a visibility arc or an arc
having an alternating sequence of length three.
To explain the rough idea of the algorithm, we skip the initialization for a moment:
assume we have an arc $\arc \in \Gamma(\sigma, p)$ and indices $L$ and $R$
such that the channel segment $\kappa_L$ is a restriction from the left of $\arc$ and we know that
no segment $\kappa_i, i \in \{1, \ldots, L-1\}$ is a violation from the left
and $\kappa_R$ is a restriction from the right with the respective property.
Then we check if $\kappa_{L+1}$ is a violation from the left, if $\kappa_{R+1}$ is a violation
from the right, if $\kappa_{L+2}$ is a violation from the left, an so forth.
As soon as we find a violation, suppose a violation from the left,
we update the index $L$ and the arc $\arc$ such that $\kappa_L$ is a restriction
from the left of $\arc$ and $\kappa_R$ remains a restriction from the right.
Then we repeat the procedure starting at the updated indices $L$ and $R$.
Initially, $\arc$ is the minimal arc in $\Gamma(\sigma, p)$ and $L = R = 0$.
We will see that this algorithm is correct as there do not appear new
\enquote{relevant} restrictions before $\kappa_L$ or $\kappa_R$, respectively.
Since we check alternately for violations from the left and from the right
and since every step can be computed in constant time the overall runtime is linear
with respect to the number of segments $n$, see Theorem~\ref{thm_linear_runtime}.\\
In Figure~\ref{figure_algorithm} the basic update steps of Algorithm~\ref{algo_klopfalgo} are illustrated.
The annotation is made with respect to the notation use in Algorithm~\ref{algo_klopfalgo}.
Note that in every step the connecting arc has an right-blocking alternating sequence of length two.
The restrictions are marked unless they are starting restrictions.
\begin{figure}[ht]
\centering
\begin{overpic}[height=3cm,angle=-33]{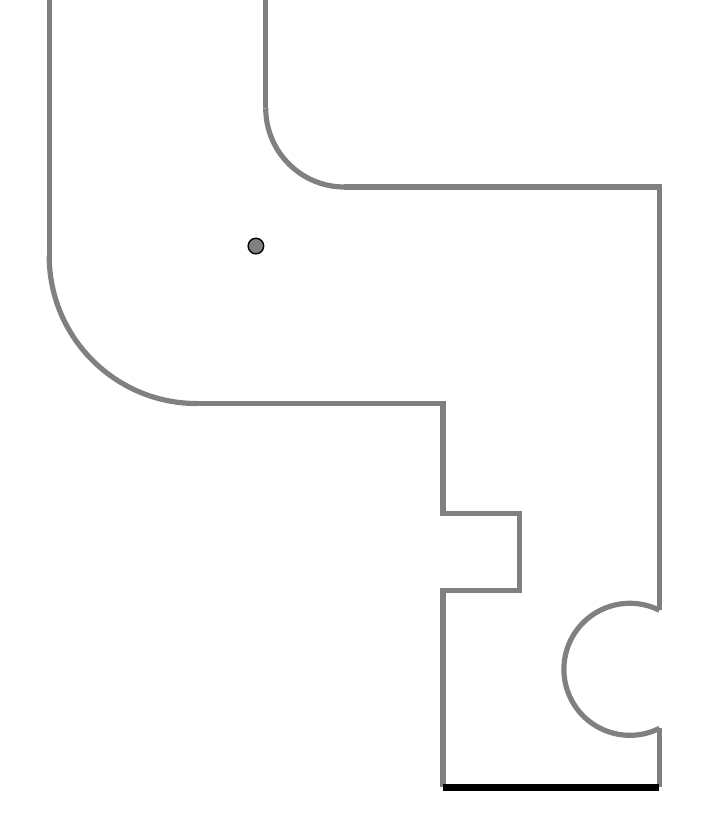}
\put(10,90){$(i)$}
\put(44,70){$p$}
\put(40,4){$\sigma$}
\end{overpic}
\begin{overpic}[height=3cm,angle=-33]{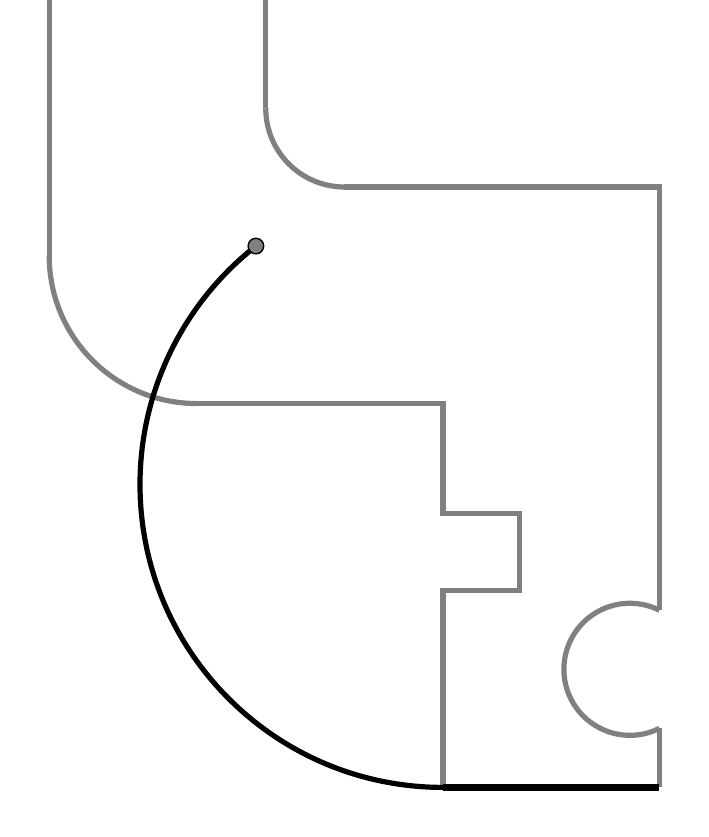}
\put(10,90){$(ii)$}
\put(44,70){$p$}
\put(40,4){$\sigma$}
\end{overpic}
\begin{overpic}[height=3cm,angle=-33]{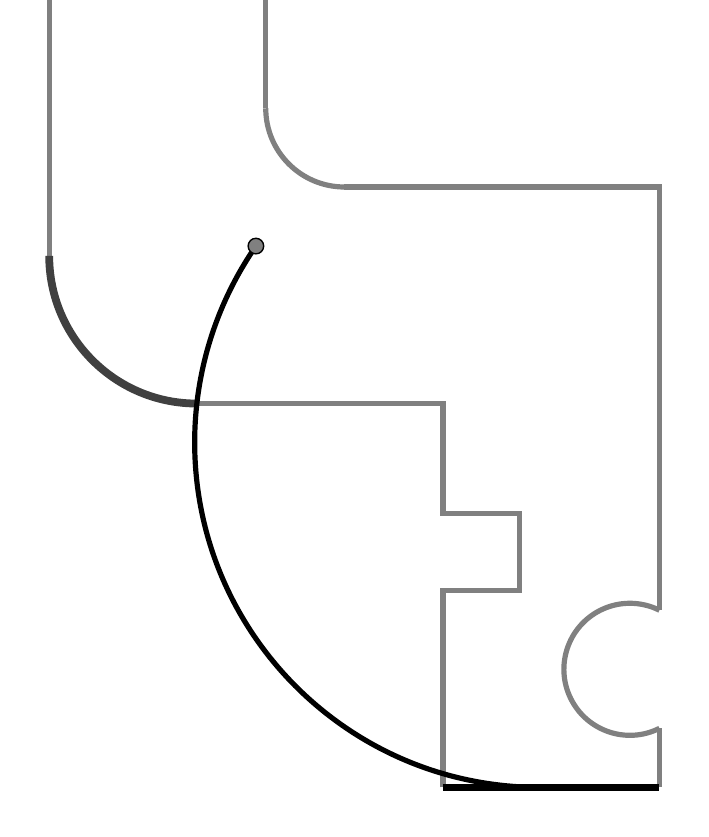}
\put(10,90){$(iii)$}
\put(44,70){$p$}
\put(40,4){$\sigma$}
\put(17,67){$\kappa_L$}
\end{overpic}\\
\vspace*{7mm}
\begin{overpic}[height=3cm,angle=-33]{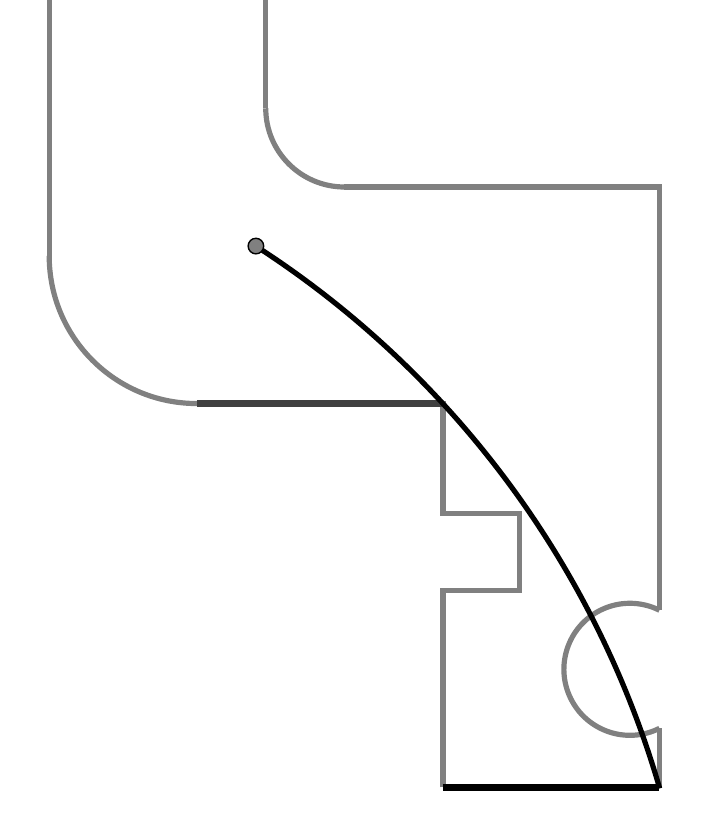}
\put(10,90){$(iv)$}
\put(44,70){$p$}
\put(40,4){$\sigma$}
\put(34,47){$\kappa_L$}
\end{overpic}
\begin{overpic}[height=3cm,angle=-33]{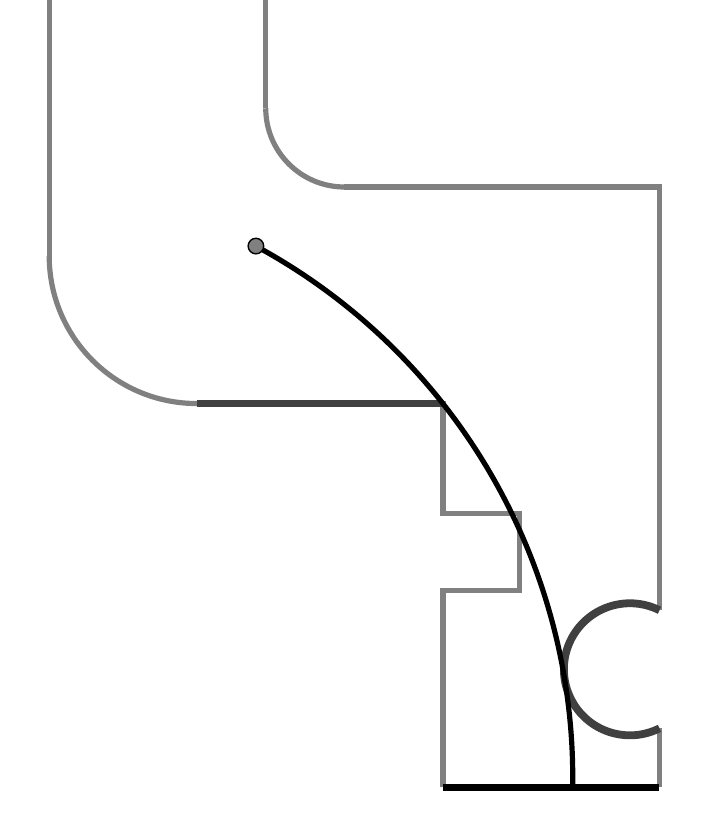}
\put(10,90){$(v)$}
\put(44,70){$p$}
\put(40,4){$\sigma$}
\put(34,47){$\kappa_L$}
\put(68,13){$\kappa_R$}
\put(67,15){\vector(-4,0){12}}%
\end{overpic}
\begin{overpic}[height=3cm,angle=-33]{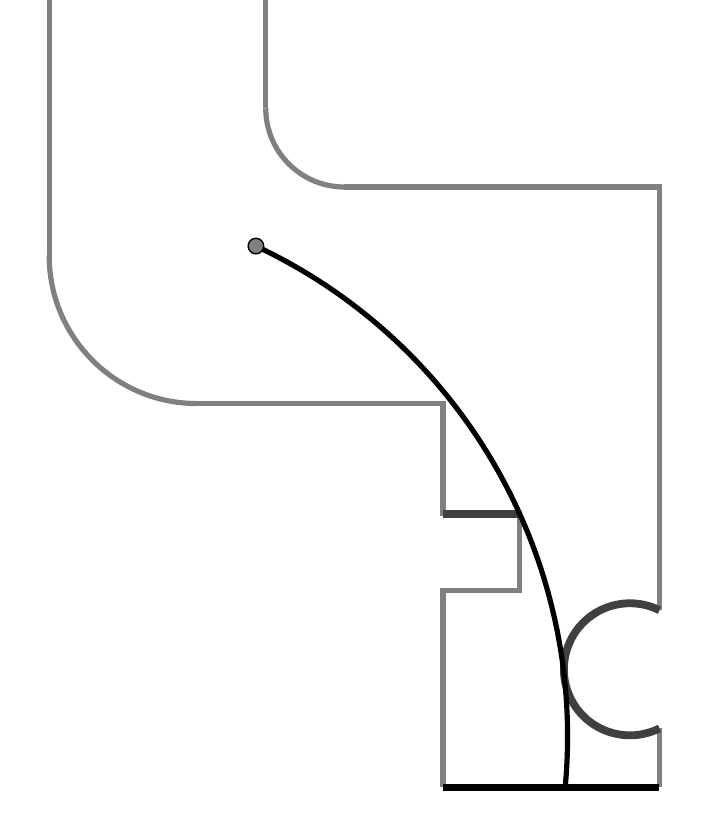}
\put(10,90){$(vi)$}
\put(44,70){$p$}
\put(40,4){$\sigma$}
\put(27,29){$\kappa_L$}
\put(38,32){\vector(4,1){12}}%
\put(68,13){$\kappa_R$}
\put(67,15){\vector(-4,0){12}}%
\end{overpic}
\caption{Illustration of the basic steps, the update steps in
Line~\ref{algo_update_arc_one} and \ref{algo_update_arc_two}, of Algorithm~\ref{algo_klopfalgo}.
\label{figure_algorithm}}
\end{figure}
\begin{rmk}
Note that in Algorithm~\ref{algo_klopfalgo} there is implicitly a third input
parameter, a tolerance up to which violations are ignored, cf.~Remark~\ref{rmk_violation_tolerance}.
To enhance readability, this tolerance is not listed and regarded explicitly in the algorithm.\\
Further, note that the Algorithm is based on continuous changes of the connecting
arc. This is the main advantage with respect to numerical stability over
any algorithm that computes arcs exactly satisfying the extremal condition.
\end{rmk}
\algrenewcommand{\algorithmiccomment}[1]{\hfill// #1}
\begin{algorithm}
\caption{}
\label{algo_klopfalgo}
\begin{algorithmic}[1]
\State \textbf{Input:} a channel with $n$ segments and a point $p$ inside the channel
\State \textbf{Output:} $\arc \in \Gamma(\sigma, p)$ that is a vis. arc or has an alt. sequ. of length three
\State Let $\arc = \min \Gamma(\sigma, p)$
\State // segment number of current restr. from the left, $L$, or right, $R$; $\kappa_0 := \sigma$
\State $L = R = 0$
\State $l = r = 0$ \Comment{current indices of segments to check for violations}
\While{$l < n$ or $r < n$} \label{algo_while}
\State Let $l = \min(l + 1, n)$; $r = \min(r + 1, n)$
\If{$\arc$ has an alternating sequence of length three with restrictions in
\hspace*{2cm}$\sigma, \kappa_L, \kappa_{L+1}, \ldots, \kappa_l, \kappa_R, \kappa_{R+1}, \ldots, \kappa_r$}\label{algo_check_one}
\State return $\arc$
\ElsIf{there is a $\arc^* \in \Gamma(\sigma, p)$ with alternating sequence of length three
\hspace*{2cm}in $\sigma, \kappa_L, \kappa_R, \kappa_l, \kappa_r$}\label{algo_check_two}
\State return $\arc^*$ \label{algo_return1}
\EndIf
\State // Update $\arc$ in case of a violation
\If{$\kappa_l$ is a violation from the left of $\arc$}\label{algo_check_violation_one}
\State Choose $\arc \in \Gamma(\sigma, p)$ as the arc with right-bl. alt. sequence $(\kappa_R, \kappa_l)$\label{algo_update_arc_one}
\State $L = l$; $r = R$ \label{algo_update_l_2} \label{algo_update_r_1}
\ElsIf{$\kappa_r$ is a violation from the right of $\arc$}\label{algo_check_violation_two}
\State Choose $\arc \in \Gamma(\sigma, p)$ as the arc with right-bl. alt. sequence $(\kappa_r, \kappa_L)$\label{algo_update_arc_two}
\State $R = r$; $l = L$ \label{algo_update_l_3} \label{algo_update_r_3}
\EndIf
\EndWhile
\State \textbf{return} $\arc$\label{algo_return_visibility_arc}
\end{algorithmic}
\end{algorithm}
\begin{thm}
Algorithm \ref{algo_klopfalgo} is correct, this means that the result is a visibility arc
or an arc in $\Gamma(\sigma, p)$ having an alternating sequence of length three.
\end{thm}
The proof is divided into two parts: first we show that the arcs in line
\ref{algo_update_arc_one} and \ref{algo_update_arc_two} exist.
This will be done in Lemma~\ref{lemma_arcs_exist}.
Afterwards, we prove an invariant of the algorithm, see Lemma~\ref{lemma_invariant},
which yields that the arc in line \ref{algo_return_visibility_arc} has no violations
and thus is a visibility arc, see Lemma~\ref{lemma_result_is_visibility_arc}.
\begin{lem}\label{lemma_arcs_exist}
The arcs in line \ref{algo_update_arc_one} and \ref{algo_update_arc_two} of
Algorithm \ref{algo_klopfalgo} exist and are unique.
\end{lem}
\begin{proof}
We assume that for any three distinct points $p_1, p_2, p_3$ there exists an arc
with starting point $p_1$, endpoint $p_3$ that passes $p_2$. This can be done
since the maximal length of relevant arcs can be limited, cf.~Remark~\ref{remark_arcs_exist}.\\
Furthermore, we assume that $p \in P^\circ$ is strictly left of $\sigma$,
which implies that the only intersection of a
$\arc \in \Gamma(\sigma, p)$ and $\sigma$ is $\arc(0)$.
If $p$ is not strictly left of $p$ then the proof works analogously,
but it has to be considered that $\arc$ and $\sigma$ might have two intersections.
Furthermore, we only consider line \ref{algo_update_arc_one}
and we assume $R \not= 0$ as the other cases can be proven similarly.\\
The proof is divided into four parts: First, we define an indexed family
of connecting arcs that represent the process of pushing $(1)$. Then we show
that $\kappa_R$ is a restriction from the right of every connecting arc
of this indexed family $(2)$ and that there is one connecting arc such that
$\kappa_l$ is a restriction from the left or such that it has an alternating
sequence of length three $(3)$. Finally, we show that the computed arc is unique $(4)$.\\[3mm]
\underline{Part 1:}
Let $\arc_1$ denote the arc $\arc$ computed so far in the algorithm and
let $t_1 := t_\sigma(\arc_1(0))$. We know that $t_1 \not= 0$ as $\arc_1$
does not have an alternating sequence of length three.
Consider the function $f : [0, t_1] \to \Gamma$
where $f(t) = \arc_t$ is the unique arc with starting point $\sigma(t)$, endpoint $p$
such that $\kappa_R \cap \arc_t \not= \emptyset$ and $\kappa_R$ is locally
right of $\arc_t$ at every $x \in \kappa_R \cap \arc_t$.
Let $t_2 := {\min\{t \in [0, t_1] \fdg f(t') \in \Gamma(\sigma, p) \text{ for all } t' \in [t, t_1] }\}$
and $\arc_2 := f(t_2)$.
As $f$ is continuous and $f(t_1) = \arc_1 \in \Gamma(\sigma, p)$, we know that $t_2$ is well defined.\\[3mm]
\underline{Part 2:} Now we show that $\kappa_R$ is a restriction from the right of every
$\arc_t, t \in [t_2, t_1)$ and that $\arc_1 < \arc_t$:
since $\arc_1 \cap \kappa_R \not= \emptyset$ and $\arc_t \cap \kappa_R \not= \emptyset$
and $\kappa_R$ is locally right of both $\arc_1$ and $\arc_t$
at every intersection point it is easy to see
that $\arc_t \capl \arc_1 \not= \emptyset$, which yields $\arc_1 < \arc_t$.
Furthermore, with $\arc_t \capl \arc_1 := \{\tau_t\}$ and $\arc_1 \capr \arc_t := \{\tau_1\}$,
i.e. $\arc_1(\tau_1) = \arc_t(\tau_t)$, we know that
$\kappa_R \cap \arc_t\restr{[0,\tau_t)} = \kappa_R \cap \arc_1\restr{(\tau_1, 1]} = \emptyset$.
Let $q \in \kappa_R$ be a restriction from the right of $\arc_1$.
By Lemma~\ref{lemma_delta_difference_on_arcs}, we get
$\Delta_q(\arc_t, \kappa) = 1$ if $q = \arc_1(\tau_1)$
and $\Delta_q(\arc_t, \kappa) = 0$ if $q = \arc_1\restr{[0,\tau_1)}$.
Hence, as $\kappa_R$ is locally right of $\arc_t$ at every point in $\kappa_R \cap \arc_t$,
in either case we know that for every $q \in \kappa_R$ we have
$\Delta_q(\arc_t, \kappa) \in \{0, 1\}$ with value $1$ if and
only if $q \in \arc_t$. This yields that $\kappa_R$ is a restriction from the right of $\arc_t$.
Analogously, one can show that every $q\in \kappa_l$ with $\Delta_q(\arc_1, \kappa) = -2$
is a restriction from the left of $\arc_t, t\in[t_2, t_1)$ if $q \in \arc_t$.\\[3mm]
\underline{Part 3:} We know that there are $q_1 \in \kappa_R, q_2 \in \kappa_l$ such that $(q_1, q_2)$ is
a right-blocking alternating sequence of $\arc_1$ as otherwise there would be an
alternating sequence of length three in $\kappa_l, \kappa_R, \kappa_L$.
Since $f$ is continuous and $\kappa_l$ and $\kappa_R$ are closed, we know that there is
a minimal $t^* \in [t_2, t_1]$ such that $\arc_4 := f(t^*)$ has
an right-blocking alternating sequence $(q_1, q_2)$ with $q_1 \in \kappa_R$, $q_2 \in \kappa_l$.
First, we show that in the case $t^* = t_2$ the arc $\arc_2 = \arc_4$
has an alternating sequence of length three since it has a starting restriction from the left.
An example showing this case is illustrated in Figure~\ref{figure_pushing} (left).
Since we assumed that $p$ is left of $\sigma$, the only way to lose the
connecting arc property is that the arc does not leave $\sigma$ in $0$ to the left.
By continuity of $\sigma' : [0,1] \to S^1$ and $(f(t))'(0) : [0,1] \to S^1$,
we know that either $t_2 = 0$, $\sigma'(t_2) = -\gamma'_2(0)$ or $\sigma'(t_2) = \gamma'_2(0)$.
If $\sigma'(t_2) = -\gamma'_2(0)$ then $\sigma$ is left of $\arc_2$. As
we know that $\arc_1 \capr \arc_2 \not= \emptyset$,
with $\arc_1 \capr \arc_2 := \{\tau\}$, there must be a $\tau' \in (0, \tau)$
such that $\arc_1$ cuts $[\arc_2]$ in $\tau'$ from the left.
So, with $p \in \arc_1 \cap \arc_2$ this would yield at least three intersections of
$\arc_1$ and $[\arc_2]$, which is a contradiction.
Hence, we know that $\gamma_2$ has a starting restriction from the left
as $t_2 = 0$ or $\sigma'(t_2) = \gamma'_2(0)$.\\
If $t^* \not= t_2$ then we know that the above property can only be lost
if the touching point on $\kappa_R$ \enquote{jumps} such that it is
after the one on $\kappa_l$ with respect to the respective arc or the intersection
with $\kappa_l$ that is after $\kappa_R$ with respect to the respective arc disappears.
In the first case, we know that 
either $(\kappa_R(0), q_2, \kappa_R(1))$ or $(\kappa_R(1), q_2, \kappa_R(0))$
is a right-blocking alternating sequence of $\arc_4$, cf. Figure~\ref{figure_pushing} in the middle.
Otherwise, we know that $\kappa_l$ is locally left of $\arc_4$ at $q_2$.
So, $\kappa_l$ is a restriction from the left of $\arc_4$ or there are
$q_3 \in \kappa_l$, $q_4 \in \kappa_R$,
such that $(q_3, q_4, \kappa_l(0))$ or $(q_3, q_4, \kappa_l(1))$
is a left-blocking alternating sequence of $\arc_4$, cf. Figure~\ref{figure_pushing} (right).
We know that $\arc_4$ does not have an alternating sequence of length
three with restrictions in $\kappa_l, \kappa_R$. Hence, $\kappa_l$ must be a
restriction from the left of $\arc_4$, so we found an arc
with right-blocking alternating sequence $(\kappa_R, \kappa_l)$.
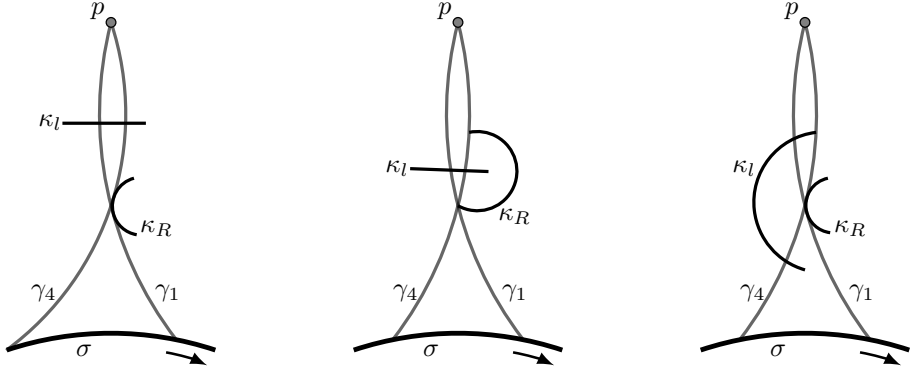
\begin{figure}[bt]
\centering
\begin{overpic}[height=5cm]{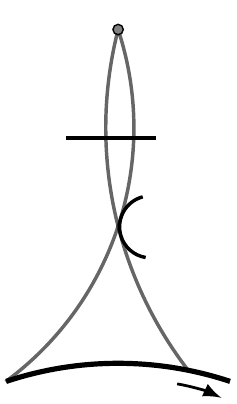}
\put(24,95){$p$}
\put(20,4){$\sigma$}
\put(41,20){$\arc_1$}
\put(8,20){$\arc_4$}
\put(37,37){$\kappa_R$}
\put(10,65){$\kappa_l$}
\end{overpic}
\hfill
\begin{overpic}[height=5cm]{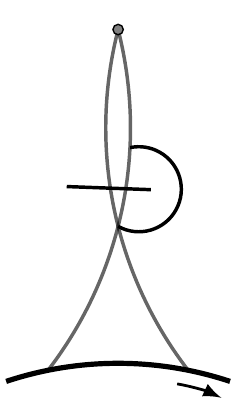}
\put(24,95){$p$}
\put(20,4){$\sigma$}
\put(41,20){$\arc_1$}
\put(12,20){$\arc_4$}
\put(40,40){$\kappa_R$}
\put(10,53){$\kappa_l$}
\end{overpic}
\hfill
\begin{overpic}[height=5cm]{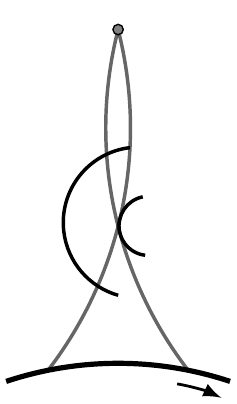}
\put(24,95){$p$}
\put(20,4){$\sigma$}
\put(41,20){$\arc_1$}
\put(12,20){$\arc_4$}
\put(37,37){$\kappa_R$}
\put(10,53){$\kappa_l$}
\end{overpic}
\caption{Illustration of some configurations appearing in the proof of Lemma~\ref{lemma_arcs_exist}.
\label{figure_pushing}}
\end{figure}
\vspace*{-1mm}

\noindent\underline{Part 4:} To show uniqueness we assume that there are two arcs
$\arc_1, \arc_2 \in \Gamma(\sigma, p)$, $\arc_1 \not=\arc_2$ with
right-blocking alternating sequence $(\kappa_R, \kappa_l)$.
Since $\kappa_R$ is a restriction from the right of $\arc_1$ and $\arc_2$, we know that
$\kappa_R \cap \arc_1 \not= \emptyset$ and $\kappa_R \cap \arc_2 \not= \emptyset$
and that $\kappa_R$ is locally right of both $\arc_1$ and $\arc_2$ at every
point of the respective intersection. Hence, we know that
$\kappa_R \subset \overline{Z}$ with $Z$ a connected component of
$\RR^2 \setminus (\arc_1 \cup \arc_2 \cup \sigma)$. Similarly, this holds for $\kappa_l$.
Distinguishing by the order of $\arc_1$ and $\arc_2$ and by the order of $\arc_1(0)$
and $\arc_2(0)$ with respect to $\prec_\sigma$ it is easy to see that this is not possible.
\qed
\end{proof}
\begin{lem}\label{lemma_invariant}
Let $q_R$ be the first point on $\kappa_R$, with respect to $\kappa$,
such that $q_R \in \arc$ and let $q_L$ be the first point on $\kappa_L$,
with respect to $\kappa$, such that $q_L \in \arc$. Then we have:
\begin{enumerate}
\item if there is a restriction from the right $q$ with $q \prec_\kappa q_R$ and
$q \prec_\arc q_R$ then, with $q$ the first such restriction with respect to $\kappa$,
we know that $\kappa$ approaches $\arc$ in $t_\kappa(q)$ from the left.
\item if there is a restriction from the left $q$ with $q \prec_\kappa q_L$ and
$q_L \prec_\arc q$ then, with $q$ the first such restriction with respect to $\kappa$,
we know that $\kappa$ approaches $\arc$ in $t_\kappa(q)$ from the right.
\end{enumerate}
\end{lem}
\begin{proof}
We only prove the first invariant, the second one can be shown analogously.
Obviously, the invariant holds initially. Consider the update step in line
\ref{algo_update_arc_two} and \ref{algo_update_l_3} and suppose $L \not= 0$
and $R \not= 0$. The update in line \ref{algo_update_arc_one} and \ref{algo_update_r_1}
and the cases with vanishing $L$ or $R$ can be proven analogously.
Let $\arc_1$ and $R_1$ be the arc $\arc$ and the value of $R$ computed
so far in the algorithm and denote the respective values after the update
by $\arc_2$ and $R_2$. Furthermore, let $r_1$ be the first point on $\kappa_{R_1}$,
with respect to $\kappa$, such that $r_1 \in \arc_1$,
let $r_2$ be the first point on $\kappa_{R_2}$ such that $r_2 \in \arc_2$ and
let $r_3$ be the first point on $\kappa_{R_2}$ such that $r_3 \in \arc_1$.\\
Assume that the invariant holds before the update and that it is violated
after the update. Then we know that there is a $q\in\kappa$ that is a
restriction from the right of $\arc_2$ and that satisfies
$q \prec_\kappa r_2$ and $q \prec_\arc r_2$. With $q$ the first
such restriction with respect to $\kappa$, we know that $\kappa$ approaches
$\arc_2$ in $t_\kappa(q)$ from the right.
Since $\arc_1 < \arc_2$ and $q \prec_{\arc_2} r_2 \prec_{\arc_2} \arc_2(t)$, with
$\{t\} := \arc_2 \capl \arc_1$, Lemma~\ref{lemma_delta_difference_on_arcs}
yields $\Delta_q(\arc_1, \kappa) = 2$. Hence, $q \prec_\kappa r_1$ as
otherwise $\arc$ and $R$ would have been updated earlier.\\
Consider the closed path
\[
\alpha := \arc_1\restr{[0,t(r_3)]} \sqcup \kappa\restr{[t(r_3), t(r_2)]}
\sqcup \overline{\arc_2\restr{[0,t(r_2)]}} \sqcup \sigma\restr{[t(\arc_2(0)), t(\arc_1(0))]}
\]
and let $Z$ be the connected component of $\RR^2 \setminus \alpha$
that is locally right of $\arc_2$ in $q$.
Obviously, $\kappa(0) \not\in Z$ and as $\kappa$ approaches
$\arc_2$ in $t_\kappa(q)$ from the right, we know that there is
an $\varepsilon > 0$ such that
$\kappa([t_\kappa(q) - \varepsilon, t_\kappa(q))) \subset Z$.
Hence, there must be a $t \in [0, t_\kappa(q))$ such that
$\kappa(t) \in \alpha$.\\
First, we consider the case $r_3 \prec_{\arc_1} r_1$. We know
that there is a $t \in [0, t_\kappa(q)]$ with
$\kappa(t) \in \arc_1\restr{[0, t(r_3))} \cup \arc_2\restr{[0, t(r_2))}$.
Let $t^* \in [0,t_\kappa(q))$ be maximal such that
$\kappa(t^*) \in \arc_1\restr{[0, t(r_3)]} \cup \arc_2\restr{[0, t(r_2)]}$
and let $q^* = \kappa(t^*)$.
Then, $\kappa\restr{(t^*, t_\kappa(q))} \subset Z$ and in particular
$\kappa\restr{(t^*, t_\kappa(q))} \cap \arc_2 = \emptyset$.
Hence, $q^* \not\in \arc_2\restr{[0, t(r_2)]}$, since otherwise
$q^*$ would be a restriction from the right, which would be
a contradiction as $q$ was defined as the first such restriction.
So, $q^* \in \arc_1\restr{[0, t(r_3)]}$. Since
$\kappa\restr{(t^*, t_\kappa(q))} \cap \arc_2 = \emptyset$ and $\kappa$ approaches
$\arc_2$ in $t_\kappa(q)$ from the right, we know that $\Delta_{q^*}(\arc_2, \kappa) = 0$.
Hence, Lemma~\ref{lemma_delta_difference_on_arcs} yields $\Delta_{q^*}(\arc_1, \kappa) = 1$,
so $q^*$ is a restriction from the right of $\arc_1$.
As $q^* \prec_\kappa r_1$ and $q^* \prec_{\arc_1} r_3 \prec_{\arc_1} r_1$, we can define $q'$ to be
the first restriction from the right of $\arc_1$ with respect to $\kappa$
such that $q' \prec_{\arc_1} r_1$ and we know $q' \prec_\kappa r_1$.
Since the invariant with respect to $\arc_1$ holds, we know that $\kappa$
approaches $\arc_1$ in $t_\kappa(q')$ from the left.
With the same arguments as before we can conclude that there must be a
$t \in [0, t_\kappa(q'))$ with
$\kappa(t) \in \arc_1\restr{[0, t(r_3)]} \cup \arc_2\restr{[0, t(r_2)]}$
and the point corresponding to the last such parameter is either
a restriction from the right of $\arc_1$ or a restriction from the right
of $\arc_2$. In either case this yields a contradiction.\\
In the case $r_1 \prec_{\arc_1} r_3$ the argumentation is basically the same.
Since\linebreak $\kappa\restr{[0,t_\kappa(q)]} \cap \kappa\restr{[t(r_1), t(r_3)]} = \emptyset$
and $\kappa\restr{[t(r_1), t(r_3)]}$ does not cut $\arc_1\restr{[t(r_1), t(r_3)]}$,
we also know that there is a $t \in [0, t_\kappa(q))$ with
$\kappa(t) \in \arc_1\restr{[0, t(r_1)]} \cup \arc_2\restr{[0, t(r_2)]}$.
Let $t^* \in [0, t_\kappa(q))$ be maximal such that
$\kappa(t^*) \in \arc_1\restr{[0, t(r_1)]} \cup \arc_2\restr{[0, t(r_2)]}$
and let $q^* = \kappa(t^*)$.
The only difference to the case $r_3 \prec_{\arc_1} r_1$ is that it not obvious
that $q^*$ is a restriction from the right of $\arc_1$ or $\arc_2$ if
$q^* \in \arc_1\restr{[0, t(r_1)]}$ or $\arc_2\restr{[0, t(r_2)]}$, respectively.
We only show that $q^*$ is a restriction from the right of $\arc_2$ if
${q^* \in \arc_2\restr{[0, t(r_2)]}}$ as the other cases can be proven analogously.
We know that there is no restriction from the left in
$\kappa\restr{[t_\kappa(r_1), t_\kappa(r_3)]} \cap \arc_1\restr{[t_1(r_1), t_1(r_3)]}$
and we assume $\kappa\restr{(t_\kappa(r_1), t_\kappa(r_3))} \cap \arc_1\restr{(t_{\arc_1}(r_1),t_{\arc_1}(r_3))} = \emptyset$,
otherwise $\kappa\restr{[t_\kappa(r_1), t_\kappa(r_3)]}$ can be split such that
there is no intersection at each part and the arguments work for every part. Let
\[
\alpha_2 := \kappa\restr{[t_\kappa(r_1), t_\kappa(r_3)]}
\sqcup \overline{\arc_1\restr{[t_{\arc_1}(r_1),t_{\arc_1}(r_3)]}}
\]
and let $Z_2$ be the connected component locally right of
$\arc_1\restr{[t_{\arc_1}(r_1),t_{\arc_1}(r_3)]}$.
Then we know that $\kappa\restr{(t^*, t(q))} \subset (Z \cup Z_2 \cup \arc_1\restr{[t_{\arc_1}(r_1),t_{\arc_1}(r_3)]})$.
Since $\Delta(\arc_2, \beta) = 0$ for every arc spline $\beta \subset Z$
and since we can split $\kappa\restr{[t^*, t_\kappa(q)]}$ such that every
part is either in $Z \cup \arc_1\restr{[t_{\arc_1}(r_1),t_{\arc_1}(r_3)]}$ or in
$Z_2 \cup \arc_1\restr{[t_{\arc_1}(r_1),t_{\arc_1}(r_3)]}$
it is enough to show that $\Delta(\arc_2, \beta) = 0$ for any arc spline
$\beta$ with $\beta(0), \beta(1) \in \arc_1\restr{(t_{\arc_1}(r_1),t_{\arc_1}(r_3))}$
and $\beta^\circ \in Z_2$.
We assume that $Z_2$ is the interior of $\alpha_2$, otherwise it works analogously.
As $\alpha_2$ leaves $\arc_1$ in $0$ to the right, we know that there is a $\varepsilon > 0$
such that $\arc_1\restr{[t_{\arc_1}(r_1) - \varepsilon, t_{\arc_1}(r_1))} \cap Z_2 = \emptyset$.
For every $\tilde{q} \in \kappa\restr{[t_\kappa(r_1), t_\kappa(r_3)]}$ we know
that $\tilde{q}$ is not a violation from the right of $\arc_1$
and if $\tilde{q} \in \arc_1\restr{[0, t_{\arc_1}(r_1)]}$
then it is not a restriction from the left of $\arc_1$. This yields
$\Delta(\arc_1\restr{[0, t_{\arc_1}(r_1) - \varepsilon]}, \alpha_2) = 0$.
Hence, with Proposition~\ref{prop_delta_vanishes},
we know that $\arc_1(0) \not\in Z_2$ and $\arc_1(1) \not\in Z_2$
since $\Delta(\arc_1, \alpha_2) = 0$.
Now, let $\beta$ be an arc spline with $\beta(0), \beta(1) \in \arc_1\restr{(t_{\arc_1}(r_1),t_{\arc_1}(r_3))}$
and $\beta^\circ \in Z_2$. Then, the interior of $\beta$, denoted by $I_\beta$, is a subset of
$Z_2$. Hence, $\arc_1(0), \arc_1(1) \not\in I_\beta$ and with Proposition~\ref{prop_delta_vanishes}
we get $\Delta(\arc_1, \beta) = 0$. Lemma~\ref{lemma_delta_difference_on_arcs} yields $\Delta(\arc_2, \beta) = 0$
as $\beta$ is supposed to be a part of the channel boundary. \qed
\end{proof}
\begin{lem}\label{lemma_result_is_visibility_arc}
The arc $\arc$ in line \ref{algo_return_visibility_arc} is a visibility arc.
\end{lem}
\begin{proof}
Let $t^*\in[0,1]$ maximal such that there is no violation in $\kappa\restr{[0,t^*]}$.
Suppose $t^* < 1$. Then we know that $\kappa(t^*)$
is a restriction of $\arc$. We suppose that $\kappa(t^*)$ is
a restriction from the right of $\arc$ as it works analogously, otherwise.
Let $q$ be the first restriction from the right of $\arc$, with respect to $\kappa$.
We know that $\kappa$ approaches $\arc$ in $t_\arc(q)$ from the right.
With $q_r$ the first restriction from the right in $\kappa_R$, we know that
$q \prec_\kappa \kappa(t^*) \prec_\kappa q_r$ which yields $q_r \prec_\arc q$ because of Lemma~\ref{lemma_invariant}.
Consider the closed arc spline $\alpha := \kappa\restr{[0, t_\kappa(q)]} \sqcup
\reverse{\arc\restr{[0,t_\arc(q)]}} \sqcup \sigma\restr{[t_\sigma(\arc(0)), 1]}$
and the connected component $Z$ of $\RR^2 \setminus \alpha$
that is locally right of $\arc\restr{[0,t_\arc(q)]}$.
We know that there is no $t \in [t_\kappa(q_r), 1]$ such that $\kappa$
leaves $\arc\restr{[0,t_\arc(q)]}$ in $t$ to the left as otherwise
there would be a violation from the right.
This yields $\kappa\restr{[t_\kappa(r), 1]} \subset \overline{Z}$
as $\kappa$ is simple and $\kappa^\circ \cap \sigma = \emptyset$.
If $\arc(0) \not= \sigma(0)$ then obviously $\kappa(1) = \sigma(0) \not\in \overline{Z}$
and we get a contradiction. Otherwise, if $\arc(0) = \sigma(0)$ then
$\Delta(\arc, \kappa) = 1$, a contradiction to Lemma~\ref{lemma_restriction_well_defined}.
\qed
\end{proof}
\begin{thm}\label{thm_linear_runtime}
Algorithm \ref{algo_klopfalgo} has linear complexity with respect to $n$,
the number of segments of the channel boundary.
\end{thm}
\begin{proof}
To reach linear runtime we have to keep a record of $\Delta_{\kappa_l(0)}(\arc, \kappa)$ and
$\Delta_{\kappa_r(0)}(\arc, \kappa)$. Initially the value is clear.
If $\arc$ is not updated in an iteration then $\Delta_{\kappa_l(0)}(\arc, \kappa)$ can easily be updated
as $\Delta_{\kappa_l(0)}(\arc, \kappa) = \Delta_{\kappa_{l-1}(0)}(\arc, \kappa) + \Delta(\arc, \kappa_{l-1})$.
If $\arc$ is updated then the value of $\Delta_{\kappa_l(0)}(\arc, \kappa)$ is known
as $\kappa_l$ is a restriction from the left.
All this operations can be computed in constant time.
The same holds for $\Delta_{\kappa_r(0)}(\arc, \kappa)$, respectively.\\
As $\kappa_L$ and $\kappa_R$ are restrictions, with Lemma~\ref{lemma_delta_difference_on_arcs}
we can compute $\Delta_q(\arc^*, \kappa)$ with $\arc^* \in \Gamma(\sigma, p)$,
$q \in \kappa_L, \kappa_R, \kappa_l, \kappa_r$ only considering the respective segment.
With $\lambda$ a channel segment, $\Delta_{\lambda(t)}(\arc, \kappa)$
as a function of $t$ is locally constant for every
$t \in [0,1]$ with $\lambda(t) \not\in \arc$.
This yields that the value changes at no more
than two intersection points, hence $\{\Delta_q(\arc, \kappa) \fdg q \in \lambda\}$,
$\lambda = \kappa_l, \kappa_r, \kappa_L, \kappa_R$ can be computed in constant time
which yields that the statements in line \ref{algo_check_violation_one} and
\ref{algo_check_violation_two} can be computed in constant time.
The check in line \ref{algo_check_two} and the arcs in line
\ref{algo_update_arc_one} and \ref{algo_update_arc_two} can also be computed
in constant time using the problem of Apollonius, cf.~\cite{coxeter1968}.
To compute line \ref{algo_check_one} in constant time, it is enough to
keep a record of the restriction from the left of $\arc$ in
$\kappa_L, \kappa_{L+1}, \ldots, \kappa_l, \kappa_R, \kappa_{R+1}, \ldots, \kappa_r$ that is minimal
with respect to $\arc$ and the maximal restriction from the right,
respectively. The minimal or maximal restriction can be updated
in constant time as only the current segment $\kappa_l$ or $\kappa_r$ must
be taken into account. It is obvious that the remaining statements can be computed in constant time.\\
To get a linear runtime we show that the loop starting in
line~\ref{algo_while} is passed at most $2n$ times.
Therefore, we show that in every iteration at least
one segment is finally processed.
In each iteration in that the algorithm
has not stopped, we can distinguish two cases:
neither $\kappa_l$ nor $\kappa_r$ is a violation from the left or right
or at least one of them, $\kappa_l$ or $\kappa_r$, is a violation from the left or right.
In the first case $l$ and $r$ are incremented by one.
Otherwise, if $\kappa_l$ is a violation from the left, $r$ is reset to $R$
but $L$ is incremented to $l$, so the segments $\kappa_{L+1}, \ldots, \kappa_l$
are finally processed. Analogously, if $r-R = k$ and $\kappa_r$ is a violation
from the right then this $k$ segments are finally processed.
As the algorithm stops if both $l$ and $r$ are equal to $n$
we know that there are at most $2n$ iterations.\qed
\end{proof}
\begin{rmk}
The result of Algorithm \ref{algo_klopfalgo} is either a visibility arc
or an arc that has an alternating sequence of length three and thus proves
that the considered point is not visible. Even if we consider numerical
errors, in either case we get an arc which proves that the respective
point is visible or not visible up to a certain tolerance. Hence, numerical
inaccuracies only affect situations where some uncertainty is inevitable.
\end{rmk}
\section{Numerical stability\label{section_stability}}
\begin{figure}[ht]
\begin{minipage}[b]{0.47\textwidth}
\begin{center}
\begin{overpic}[height=6.2cm]{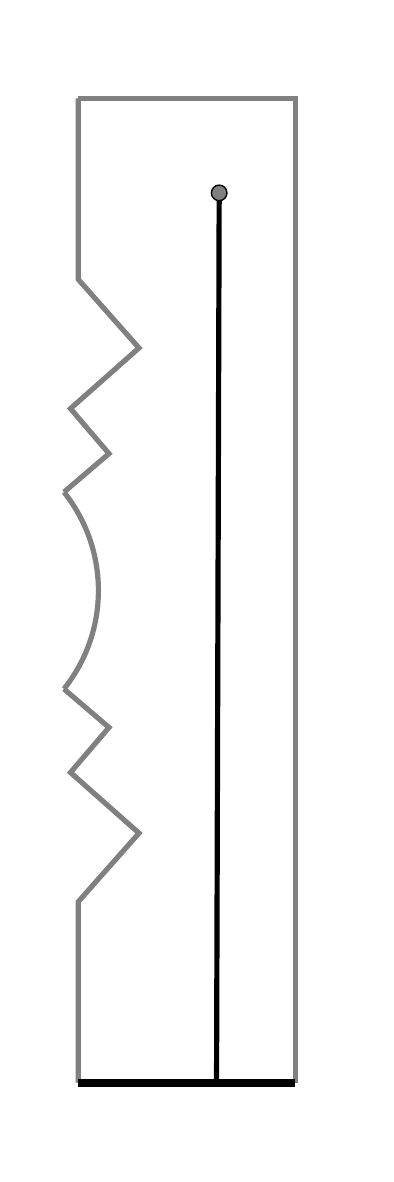}
\put(15,4){$\sigma$}
\put(0,87){$P$}
\put(14,84){$p$}
\put(20,48){$\gamma$}
\end{overpic}
\caption{Channel $P$ with starting arc $\sigma$ a visible point $p$ and a visibility arc $\gamma$
\label{figure_stability_visibility_line}}
\end{center}
\end{minipage}
\hfill
\begin{minipage}[b]{0.47\textwidth}
\begin{center}
\begin{overpic}[height=6.2cm]{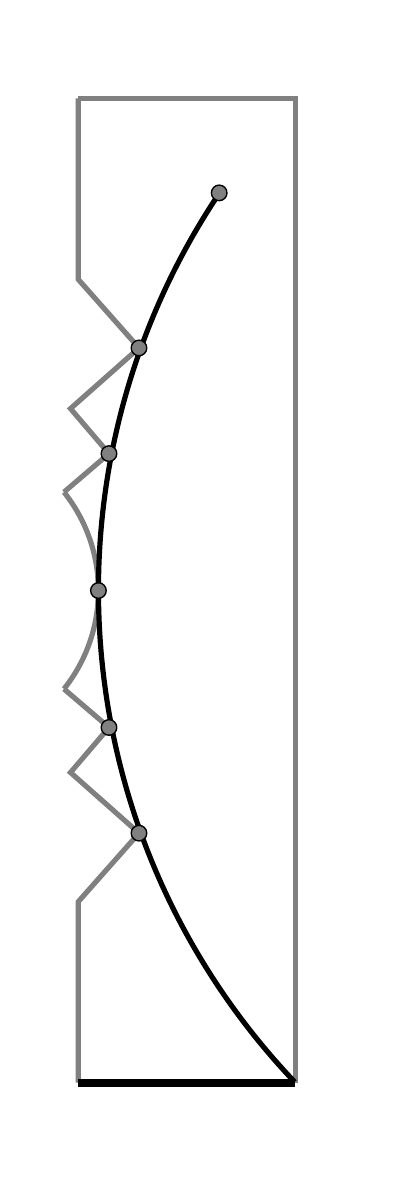}
\put(17,4){$\sigma$}
\put(0,87){$P$}
\put(21,84){$p$}
\put(17,20){$\gamma$}
\end{overpic}
\caption{A visibility arc which has several alternating sequences of length two
\label{figure_stability_visibility_arc}}
\end{center}
\end{minipage}
\end{figure}
In this section we will show how Algorithm~\ref{algo_klopfalgo} from Section~\ref{section_algorithm}
works on a numerically critical situation.
Assume we have the setting depicted in Figure~\ref{figure_stability_visibility_line}.
There is no doubt that the point $p$ is circularly visible as there is a visibility arc that
is clearly inside the channel. Furthermore, $p$ is not even close to the boundary of the visibility set
as the whole interior if the channel is circularly visible.
Nevertheless, we will now show that the computation of a visibility arc has some numerical issues.
One can show that there are exactly two visibility arcs having an alternating sequence of length two,
one with a left-blocking and one with a right-blocking alternating sequence.
It is reasonable to only regard visibility arcs having an alternating sequence
as every visibility algorithm in literature is in some way based on alternating sequences.
In fact, they are inevitable if one is interested in the boundary arcs of the
visibility set which usually the case, cf.~\cite{Maier2014}.
So, let $\gamma$ be the visibility arc that has a right-blocking alternating sequence
as depicted in Figure~\ref{figure_stability_visibility_arc}.
Although, $\gamma$ does not have an alternating sequence of length greater than two,
$\gamma$ and $\kappa$ have several points in common. So, finding the visibility arc $\gamma$
is comparable to the problem of finding a solution of an over-determined system of equations
which is highly unstable.
Every little change of the channel or the point $p$ results in a
different alternating sequence of the visibility arc. Although the actual choice
of the associated alternating sequence is not relevant to decide if a point is visible,
it can be a problem for the algorithm as they need to compute an alternating sequence.\\
Now, let us analyze Algorithm~\ref{algo_klopfalgo} in this situation.
Initially, $\gamma = \min \Gamma(\sigma, p)$ as depicted in Figure~\ref{figure_stability_algorithm} (left).
Then there are two update steps as shown in the intermediate and right example of Figure~\ref{figure_stability_algorithm}.
After the second update, we have the numerically critical situation from above where $\gamma$
is a visibility arc that has several points in common with the channel boundary.
We will see that this is not a problem in Algorithm~\ref{algo_klopfalgo}.
In the next steps of the algorithm every subsequent channel segment is checked
for a violation from the left.
Actually, there should not be violations exceeding a certain tolerance.
But even if due to numerical inaccuracies a violation from the left is found,
then $L$ is updated such that the violation turns into a restriction
but the resulting arc basically remains the same.
As only subsequent channel segments are considered it is not possible that the
former restriction is found as a violation later on, which would lead to an
inconsistency or even an infinite loop.\\
So, Algorithm~\ref{algo_klopfalgo} returns without any doubt that $p$ is visible.
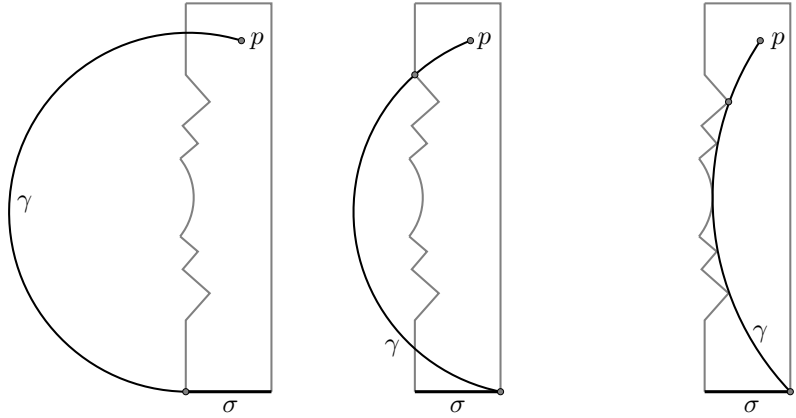
\begin{figure}[t]
\begin{minipage}[b]{0.31\textwidth}
\begin{center}
\begin{overpic}[height=6.2cm]{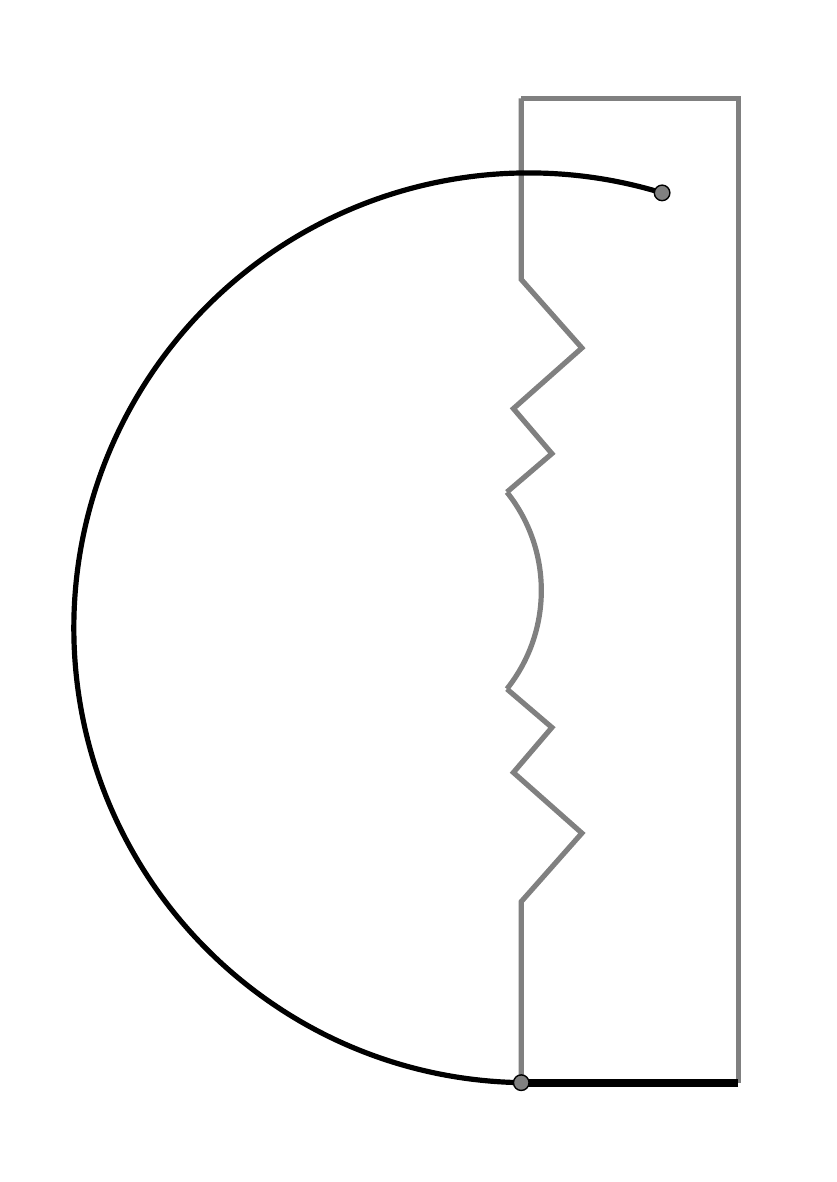}
\put(52,4){$\sigma$}
\put(58,83){$p$}
\put(8,48){$\gamma$}
\end{overpic}
\end{center}
\end{minipage}
\hfill
\begin{minipage}[b]{0.31\textwidth}
\begin{center}
\begin{overpic}[height=6.2cm]{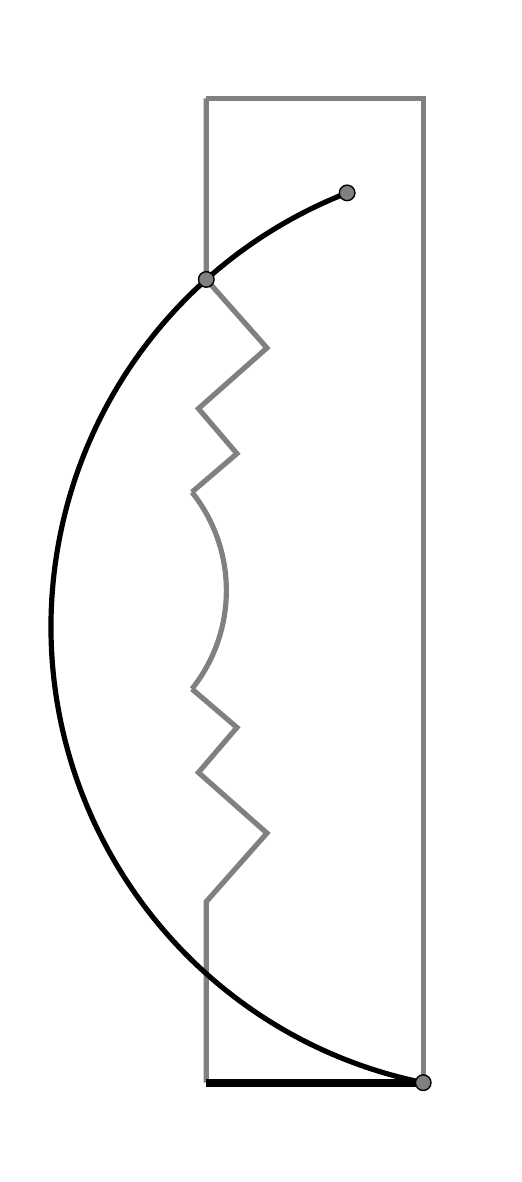}
\put(25,4){$\sigma$}
\put(31,83){$p$}
\put(11,17){$\gamma$}
\end{overpic}
\end{center}
\end{minipage}
\hfill
\begin{minipage}[b]{0.31\textwidth}
\begin{center}
\begin{overpic}[height=6.2cm]{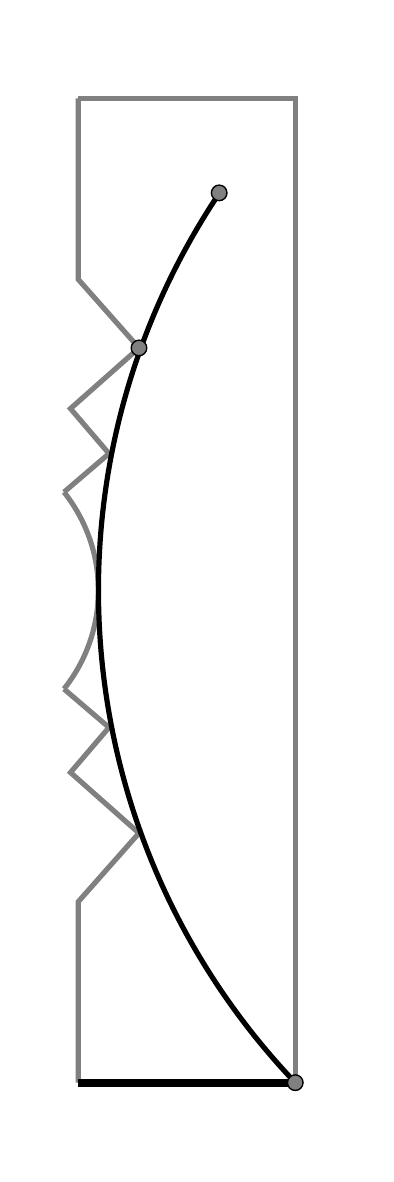}
\put(15,4){$\sigma$}
\put(21,83){$p$}
\put(17,20){$\gamma$}
\end{overpic}
\end{center}
\end{minipage}
\caption{Illustration of the steps of Algorithm~\ref{algo_klopfalgo}\label{figure_stability_algorithm}}
\end{figure}
\section{Conclusion}
We treat the problem if a point inside a simple closed
arc spline is circularly visible from a boundary arc.
In particular, we provide an easy-to-check criterion
that implies that the point is not visible
and we present a simple and numerically stable linear time
algorithm that checks visibility.\\
We will integrate the results into the SMAP approach,
see \cite{Maier2014}. Thus, point sequences can be
approximated by an arc spline up to a arbitrary tolerance
with optimal segment number in a numerically stable and
efficient manner.

\vspace*{6mm}
\noindent\textbf{Acknowledgments}
\vspace*{3mm}

The first author has been funded by the German Research Foundation
(DFG~-~Deutsche Forschungsgemeinschaft) under grant number MA 5834/1-1.
\vspace*{6mm}
\noindent\textbf{References}

\bibliography{bib_cagd}
\bibliographystyle{plain}

\end{document}